\documentclass[a4paper, UKenglish]{llncs}
\usepackage[normalem]{ulem}
\usepackage{multirow}
\usepackage[misc,geometry]{ifsym}
\usepackage{amssymb}
\usepackage{verbatim}
\usepackage{amsmath,amsfonts}
\usepackage{hyperref,cleveref}

\usepackage{graphicx}
\usepackage{color}
\usepackage{xcolor}
\usepackage{amsmath}
\usepackage{amsfonts}
\usepackage{amssymb}

\newcommand{\Ii}{\mathcal{I}}

\newcommand{\cmax}{\mathsf{cmax}}
\newcommand{\qseq}{\mathsf{qseq}}
\newcommand{\freg}{\mathsf{\mathcal{F}_I}}
\newcommand{\sreg}{\mathsf{\mathcal{P}_I}}

\newcommand{\fregm}{\mathcal{F}}
\newcommand{\sregm}{\mathcal{P}}
\newcommand{\fregk}{\mathcal{F}^{k}}
\newcommand{\sregk}{\mathcal{P}^{k}}
\newcommand{\fregnm}{\mathcal{F}^{n}_{\mathsf{I_1,\ldots, I_n}}}
\newcommand{\sregnm}{\mathcal{P}^{n}_{\mathsf{I_1,\ldots, I_n}}}
\newcommand{\fregkm}{\mathcal{F}^{k}_{\mathsf{I_1,\ldots, I_k}}}
\newcommand{\sregkm}{\mathcal{P}^{k}_{\mathsf{I_1,\ldots, I_k}}}

\newcommand{\fut}{\mathsf{F}}
\newcommand{\sbf}{\mathcal{G}}

\newcommand{\sbm}{\mathcal{H}}
\newcommand{\past}{\mathsf{P}}
\newcommand{\until}{\mathsf{U}}
\newcommand{\since}{\mathsf{S}}

\newcommand{\wB}{\mathcal{G}^w}

\newcommand{\nex}{\oplus}

\newcommand{\col}{\mathsf{Col}}
\newcommand{\normalize}{\mathsf{Norm}}
\newcommand{\Clcap}{\mathsf{CL}}

\newcommand{\pnregmtl}{\mathsf{PnEMTL}}

\newcommand{\mtl}{\mathsf{MTL}}
\newcommand{\tptl}{\mathsf{TPTL}}
\newcommand{\mitl}{\mathsf{MITL}}
\newcommand{\emitl}{\mathsf{EMITL}}

\newcommand{\ltl}{\mathsf{LTL}}

\newcommand{\emitlinf}{\mathsf{EMITL_{0,\infty}}}

\newcommand{\R}{\mathbb{R}}

\newcommand{\re}{\mathsf{A}}

\newcommand{\A}{\mathcal{A}}

\newcommand{\I}{\mathcal{I}}

\newcommand{\seq}{\mathsf{seq}}
\newcommand{\rev}{\mathsf{Rev}}
\newcommand{\Sseq}{\mathcal{T}(I_\nu)}

\newcommand{\bseq}{\mathsf{BSequence}}

\newcommand{\init}{\mathsf{init}}
\newcommand{\Qseq}{\mathsf{Qseq}}
\newcommand{\fd}{\mathsf{fd}}

\newcommand{\first}{\mathsf{first}}
\newcommand{\last}{\mathsf{last}}
\newcommand{\anch}{\mathsf{anch}}

\newcommand{\Last}{\mathsf{Last}}
\newcommand{\Lastts}{\mathsf{LastTS}}
\newcommand{\Ovs}{\mathsf{OVS}}

\newcommand{\ovs}{\mathsf{ovs}}
\newcommand{\Int}{\mathsf{Int}}
\newcommand{\nt}{\mathsf{int}}

\newcommand{\ext}{\mathsf{ext}}

\newcommand{\intintervaln}{\mathcal{I}_\mathsf{nat}}
\newcommand{\intinterval}{\mathcal{I}_\mathsf{int}}
\newcommand{\dseq}{\mathsf{dseq}}
\newcommand{\Tseq}{\mathsf{Tseq}}
\newcommand{\tseq}{\mathsf{tseq}}
\newcommand{\Dseq}{\mathsf{Dseq}}
\newcommand{\boundaryint}{\mathsf{Boundary}}

\newcommand{\oomit}[1]{}

\begin{document}

\title{Generalizing Non-Punctuality for Timed Temporal Logic with Freeze Quantifiers.}
\author{Shankara Narayanan Krishna \inst{1} \and
Khushraj Madnani\inst{2}{(\Letter)} \and
Manuel Mazo Jr. \inst{2} \and Paritosh K. Pandya \inst{1}}
\authorrunning{S.N.Krishna et al.}
\institute{IIT Bombay, Mumbai, India \email{krishnas,pandya58}@cse.iitb.ac.in\and Delft University of Technology, Delft, The Netherlands \email{k.n.madnani-1,m.mazo}@tudelft.nl}
\maketitle
\begin{abstract} 
Metric Temporal Logic (MTL) and Timed Propositional Temporal Logic (TPTL) are prominent real-time extensions of Linear Temporal Logic (LTL). In general, the satisfiability checking problem for these extensions is undecidable when both the future U and the past S modalities are used.  In a classical result, the satisfiability checking for MITL[U,S], a non-punctual fragment of MTL[U,S], is shown to be decidable with EXPSPACE complete complexity. Given that this notion of non-punctuality does not recover decidability in the case of TPTL[U,S], we propose a 
 generalization of non-punctuality called \emph{non-adjacency} for TPTL[U,S], and focus on its  1-variable fragment, 1-TPTL[U,S].  While 
non-adjacent 1-TPTL[U,S] appears to  be a very small fragment, it is  strictly more expressive than MITL. As our main result, we show that the satisfiability checking problem for non-adjacent 1-TPTL[U,S] is decidable with EXPSPACE complete complexity. 
\end{abstract}

\section{Introduction}
Metric Temporal Logic ($\mtl$) and Timed Propositional Temporal Logic ($\tptl$) are natural extensions of Linear Temporal Logic ($\ltl$) for specifying real-time properties \cite{AH94}. $\mtl$ extends the $\until$ and $\since$ modalities of $\ltl$ by associating a timing interval with these. $a \until_I b$ describes behaviours modeled as timed words consisting of a sequence of $a$'s followed by 
a $b$ which occurs at a time within (relative) interval $I$.
On the other hand, $\tptl$ uses freeze quantification  to store the current time stamp. A Freeze quantifier with clock variable $x$  has the  form $x.\varphi$. When it is 
 evaluated at a point $i$ on a timed word, the time stamp $\tau_i$ at $i$ is frozen in $x$, and the formula $\varphi$ is evaluated using this value for $x$.
 Variable $x$ is used in $\varphi$ in a constraint of the form $T-x \in I$; this constraint, when evaluated at a point $j$, checks if $\tau_j -\tau_i \in I$, where $\tau_j$ is the time stamp at point $j$.
 For example, the formula  $\fut x.(a\wedge \fut (b \wedge T-x \in [1,2] \wedge \fut (c \wedge T-x \in [1,2])))$\footnote{Here $T$ is a special symbol denoting the timestamp of the present point and $x$ is the clock that was frozen when $x.$ was asserted.}
  asserts that there is a point in future where $a$ holds and in its future within interval $[1,2]$, $b$ and $c$ occur, and the former occurs before the latter. This property is not expressible in $\mtl[\until,\since]$ \cite{Patricia,simoni}. Moreover, every property in $\mtl[\until,\since]$ can be expressed in  1-$\tptl[\until,\since]$. Thus, 1-$\tptl[\until,\since]$ is strictly more expressive than $\mtl[\until,\since]$. Unfortunately, both the logics have an undecidable satisfiability problem, making automated analysis for these logics theoretically impossible. 

Exploring natural decidable variants of these logics has been an active area of research since the advent of these logics \cite{AH93}\cite{raskin-thesis}\cite{icalp-raskin}\cite{Wilke}\cite{Rabinovich}\cite{rabin}\cite{rabinovichY}. 
One line of work restricted itself to the future only fragments $\mtl[\until]$ and 1-$\tptl[\until]$ which have both been shown to have decidable  satisfiability over finite timed words, under a pointwise interpretation\cite{Ouaknine05,OWH}. The complexity however is non-primitive recursive. 
Reducing the complexity to elementary has been challenging. One of the most celebrated of such logics is the Metric Interval Temporal Logic ($\mitl[\until,\since]$) \cite{AFH96},  a subclass of $\mtl[\until,\since]$ where the timing intervals are restricted to be non-punctual (i.e. intervals of the form $\langle x, y \rangle$ where $x < y$). The  satisfiability checking for $\mitl$ formulae is decidable with EXPSPACE complete complexity \cite{AFH96}.  While non-punctuality helps to recover the decidability of $\mtl[\until,\since]$, it does not help $\tptl[\until,\since]$.  
The freeze quantifiers  of $\tptl$ enables us to  trivially express punctual timing constraints using only the non-punctual intervals : for instance 
the 1-$\tptl$ formula $x.(a \until (a \wedge T-x \in [1, \infty) \wedge T-x \in [0, 1]))$ uses only non-punctual intervals but captures the $\mtl$ formula $a \until_{[1,1]} b$. Thus, a more refined notion of non-punctuality is needed to recover the decidability of 1-$\tptl[\until,\since]$. 

\noindent \textbf{Contributions}. With the above observations, to obtain a decidable class of 1-$\tptl[\until,\since]$ akin to $\mitl[\until,\since]$, we  revisit the notion of non-punctuality as it stands currently. As our first contribution, we propose \emph{non-adjacency}, a refined version of  non-punctuality. Two intervals, $I_1$ and $I_2$ are non-adjacent if the supremum of $I_1$ is not equal to the infimum  of $I_2$.
Non-adjacent 1-$\tptl[\until,\since]$ is the subclass of 1-$\tptl[\until,\since]$  where,  every interval used in clock constraints within the same freeze quantifier is non-adjacent to itself and to every other timing interval that appears within the same scope. (Wlog, we consider  formulae in negation normal form only.)
The non-adjacency restriction disallows punctual timing intervals : every punctual timing interval is adjacent to itself.  It can be  shown  (Theorem \ref{thm:me})  that non-adjacent 1-$\tptl[\until,\since]$, while seemingly very restrictive, is strictly more expressive than $\mitl$ and it can also express the counting and the Pnueli modalities \cite{rabin}. Thus, the logic is of considerable interest in practical real-time specification. See the full version for an example.


Our second  contribution is to give a decision procedure for the satisfiability checking of non-adjacent 1-$\tptl[\until,\since]$. We do this in two steps. {\bf 1)} We introduce a logic $\pnregmtl$ which combines and generalises the automata modalities of  \cite{raskin-thesis,Wilke,H19} and the Pnueli modalities of \cite{rabin,rabinovichY,Rabinovich}, and has not been studied before to the best of our knowledge. We show that a formula of non-adjacent 1-$\tptl[\until,\since]$ can be reduced to an equivalent formula of non-adjacent $\pnregmtl$ (Theorem \ref{thm: tptltopnregmtlsat}).  
{\bf 2)} We prove that the satisfiability of non-adjacent $\pnregmtl$ is decidable with  EXPSPACE complete complexity (Theorem \ref{thm:main}). For brevity, some of the proof details are omitted here and can be found in the full version \cite{kkmpfull}.

\noindent \textbf{Related Work and Discussion} Much of the related work has already been discussed. $\mitl$ with counting and Pnueli modalities has been shown to have EXPSPACE-complete satisfiability \cite{Rabinovich}\cite{count}. Here, we tackle even more expressive logics: namely non-adjacent 1-$\tptl[\until,\since]$ and non-adjacent $\pnregmtl$. We show that EXPSPACE-completeness of satisfiability checking is retained in spite of the additional expressive power. These decidability results are proved by equisatisfiable reductions to logic  $\emitl_{0,\infty}$ of Ho \cite{H19}. As argued by Ho, it is quite practicable to extend the existing model checking tools like UPPAAL to logic $\emitl_{0,\infty}$ and hence to our logics too. 

Addition of regular expression based modalities to untimed logics like LTL has been found to be quite useful for practical specification; even the IEEE standard temporal logic PSL has this feature. With a similar motivation, 
there has been considerable recent work on adding regular expression/automata based modalities to $\mtl$ and $\mitl$. Raskin as well as Wilke added automata modalities to $\mitl$ as well as an Event-Clock logic \emph{ECL} \cite{raskin-thesis,Wilke} and showed the decidability of satisfaction. The current authors showed that $\mtl[\until,\since_{NP}]$ (where $\until$ can use punctual intervals but $\since$ is restricted to non-punctual intervals), when extended with counting as well as regular expression modalities  preserves decidability of  satisfaction \cite{time14,KKP17,Kkp18,khushraj-thesis}. Recently,  Ferr{\`{e}}re showed the EXPSPACE decidability of MIDL which is $\ltl[\until]$ extended with a fragment of timed regular expression modality \cite{F18}. Moreover, Ho has investigated a PSPACE-complete fragment $\emitl_{0,\infty}$ \cite{H19}. Our non-adjacent $\pnregmtl$ is a novel extension of MITL with modalities which combine the features of EMITL  \cite{raskin-thesis,Wilke,H19} and the Pnueli modalities\cite{rabin,rabinovichY,Rabinovich}. 

\section{Preliminaries}
\label{sec:prelim}
  Let $\Sigma$ be a finite set of propositions, and let  $\Gamma = 2^{\Sigma} \setminus \emptyset$. 
   A word over $\Sigma$ is a finite sequence $\sigma = \sigma_1 \sigma_2 \ldots \sigma_n$, where $\sigma_i \in \Gamma$.
  A timed word $\rho$ over $\Sigma$ is a finite sequence of pairs $(\sigma, \tau)  \in \Sigma \times \R_{\geq 0}$;  
$\rho = (\sigma_1, \tau_1) \ldots (\sigma_n, \tau_n) \in (\Sigma \times \R_{\geq 0})^*$ 
where $\tau_1=0$ and $\tau_i \leq \tau_j$ for all $1 \leq i \leq j \leq n$.  The $\tau_i$ are called time stamps. 
For a timed or untimed word $\rho$, let $dom(\rho) = \{i | 1 \le i \le |\rho|\}$ where $|\rho|$ denotes the number of (event, timestamp) pairs composing the word $\rho$, and 
$\sigma[i]$ denotes the symbol at position $i \in dom(\rho)$.
The set of timed words over $\Sigma$ is denoted $T\Sigma^*$.
Given a (timed) word $\rho$ and $i \in dom(\rho)$, a pointed (timed) word is the pair $\rho, i$. 
Let $\intinterval$ ($\intintervaln$)  be the set of open, half-open or closed time intervals, such that the end points of these intervals are in $\mathbb{Z}\cup \{-\infty,\infty\}$ ($\mathbb{N} \cup \{0,\infty\}$, respectively).
We assume familiarity with $\ltl$.

\smallskip 

\noindent \textbf{Metric Temporal Logic($\mtl$)}.
$\mtl$ is a real-time extension of $\ltl$ where the modalities ($\until$ and $\since$) are guarded with intervals. Formulae of $\mtl$ are built from $\Sigma$ using Boolean connectives and 
time constrained versions  $\until_I$ and $\since_I$ of the standard $\until,\since$ modalities, where 
 $I \in \intintervaln$. 
Intervals of the form $[x,x]$  are called  punctual; a non-punctual interval is one which is not punctual. Formulae in $\mtl$ are defined as follows.  
$\varphi::=a ~|\top~|\varphi \wedge \varphi~|~\neg \varphi~|
~\varphi \until_I \varphi ~|~ \varphi \since_I \varphi$, 
where $a \in \Sigma$ and  $I \in \intintervaln$.    
For a timed word $\rho = (\sigma_1, \tau_1 ) (\sigma_2, \tau_2) \ldots (\sigma_n, \tau_n) \in T\Sigma^*$, a position 
$i \in dom(\rho)$, an $\mathsf{MTL}$ formula $\varphi$, the satisfaction of $\varphi$ at a position $i$ 
of $\rho$, denoted $\rho, i \models \varphi$, is defined below. We discuss the time constrained modalities.  \\
 \noindent $\bullet$ $\rho,i\ \models\ \varphi_{1} \until_{I} \varphi_{2}$  iff  $\exists j > i$, 
$\rho,j\ \models\ \varphi_{2}, \tau_{j} - \tau_{i} \in I$, and  $\rho,k\ \models\ \varphi_{1}$ $\forall$ $i< k <j$, 
\noindent $\bullet$ $\rho,i\ \models\ \varphi_{1} \since_{I} \varphi_{2}$  iff  $\exists j < i$, 
$\rho,j\ \models\ \varphi_{2}, \tau_{j} - \tau_{i} \in I$, and  $\rho,k\ \models\ \varphi_{1}$ $\forall$ $j< k <i$.\\ 
The language of an $\mtl$ formula $\varphi$ is defined as $L(\varphi) = \{\rho | \rho, 1 \models \varphi\}$.
Using the above, we obtain some derived formulae : the \emph{constrained eventual} operator $\fut_I \varphi \equiv true \until_I \varphi$ and its dual is  $\sbf_I \varphi \equiv \neg \fut_I \neg \varphi$. Similarly $\sbm_I \varphi \equiv true \since_I \varphi $. The  \emph{next} operator is defined as  $\nex_I \varphi \equiv \bot \until_I \varphi$. The non-strict versions of $\fut, \sbf$ are respectively denoted $\fut^w$ and $\wB$ 
and include the present point. Symmetric 
non-strict versions for past operators are also allowed. 
The subclass of $\mathsf{MTL}$ obtained by restricting the intervals $I$ in the until and since 
modalities to {\bf non-punctual intervals} is denoted  $\mathsf{MITL}$. We say that a formula $\varphi$ is satisfiable iff $L(\varphi)\neq \emptyset$. 
\begin{theorem}
 \label{thm-basic}
	 Satisfiability checking for $\mtl[\until,\since]$ is undecidable \cite{AH93}.Satisfiability Checking for $\mathsf{MITL}$  is EXPSPACE-complete\cite{AFH96}.
\end{theorem}
\noindent \textbf{Time Propositional Temporal Logic ($\tptl$)}. The logic $\tptl$ also extends $\ltl$ using freeze quantifiers. 
Like $\mtl$, $\tptl$ is also evaluated on timed words. 
Formulae of $\mathsf{TPTL}$ are built from $\Sigma$  using Boolean connectives, modalities 
 $\until$ and $\since$ of $\ltl$. In addition, $\tptl$ uses a finite set of 
  real valued clock variables  $X = \{x_1,\ldots,x_n\}$. Let $\nu: X \rightarrow \mathbb{R}_{\geq 0}$ represent  a
  valuation assigning a non-negative real value to each clock variable. 
    The formulae of $\tptl$ are defined as follows. Without loss of generality  we work with $\tptl$ in the negation normal form. 
    $\varphi::=a~|~\neg a ~|\top~|~\bot~|~ x.\varphi ~|~ T-x \in I ~|~x-T \in I~|~\varphi \wedge \varphi~|~ \varphi \vee \varphi~|
~\varphi \until \varphi ~|~ \varphi \since \varphi ~|~ \sbf \varphi ~|~ \sbm \varphi$, 
where $x \in X$, $a \in \Sigma$, $I \in \intinterval$. Here  $T$ denotes the time stamp of the point where the formula is being evaluated.
 $x. \varphi$ is the freeze quantification construct which remembers the time stamp of the current point in variable $x$ and evaluates $\varphi$. 

For a timed word $\rho=(\sigma_1,\tau_1)\dots(\sigma_n,\tau_n)$, $i \in dom(\rho)$ and a $\tptl$ formula $\varphi$, we define the satisfiability relation, $\rho, i, \nu \models \varphi$ with valuation $\nu$ of all the clock variables. We omit the semantics of Boolean, $\until$ and $\since$ operators as they are similar to those of $\ltl$. 

	   \noindent $\bullet$ $\rho, i, \nu \models a$   iff  $a \in \sigma_{i}$, and 
			$\rho,i,\nu \models x.\varphi $  iff $\rho,i,\nu[x \leftarrow \tau_i] \models \varphi$\\
	\noindent $\bullet$	$\rho,i,\nu \models T-x\ \in I $  iff $\tau_i - \nu(x) \in I$, and $\rho,i,\nu \models x-T \in I $   iff  $\nu(x) - \tau_i \in I$\\
		\noindent $\bullet$ $\rho,i,\nu\ \models\ \sbf \varphi$   iff  $\forall j > i$, $\rho,j,\nu \ \models\ \varphi$, and\\
	\noindent $\bullet$	$\rho,i,\nu\ \models\ \sbm \varphi$   iff    $\forall j < i$, $\rho,j,\nu \ \models\ \varphi$

Let $\overline{0}=(0,0,\dots,0)$ represent the initial valuation of all clock variables. 	
For a timed word $\rho$ and $i \in dom(\rho)$, we say that $\rho, i$ satisfies $\varphi$ denoted $\rho, i \models \varphi$ iff $\rho,i,\overline{0}\models \varphi$. The language of $\varphi$, $L(\varphi) = \{\rho | \rho, 1 \models \varphi\}$. The Pointed Language of $\varphi$ is defined as $L_{pt}(\varphi) = \{\rho,i | \rho, i \models \varphi\}$. Subclass of $\tptl$ that uses {\bf only 1 clock variable} (i.e. $|X| = 1$) is known as 1-$\tptl$.
The satisfiability checking for 1-$\tptl[\until, \since]$ is undecidable, which is implied by theorem \ref{thm-basic} and the fact that 1-$\tptl[\until, \since]$ trivially generalizes $\mtl[\until,\since]$ . 
\smallskip 
As an example, the formula $\varphi=x.(a \until (b \until (c \wedge T-x \in [1,2])))$ 
is satisfied by the timed word $\rho=(a,0)(a,0.2)(b,1.1)(b,1.9)(c,1.91)(c,2.1)$ since $\rho, 1 \models \varphi$.  
The word $\rho'=(a,0)(a,0.3)(b,1.4)(c,2.1)(c,2.5)$ does not satisfy $\varphi$. 
However, $\rho',2 \models \varphi$ : if we start from the second position of $\rho'$, 
we assign $\nu(x)=0.3$, and when we reach the position 4 of $\rho'$ with $\tau_4=2.1$
we obtain $T-x=2.1-0.3 \in [1,2]$. 

\section{Introducing Non-Adjacent 1-$\tptl$ and $\pnregmtl$}
\label{sec:natptl}
In this section, we define 
non-adjacent 1-$\tptl$ and $\pnregmtl$ logics. Let $x$ denote the unique freeze variable we use in 1-$\tptl$.
\\\noindent\textbf{Non-Adjacent 1-$\tptl$} is defined as a subclass of 1-$\tptl$ where adjacent intervals within the scope of any freeze quantifier is disallowed. Two intervals $I_1, I_2 \in \intinterval$ are non-adjacent 
iff $\sup(I_1) {=} \inf(I_2) \Rightarrow \sup(I_1) {=} 0$.
A set $\mathcal{I}_{na}$ of intervals is non-adjacent iff any two intervals 
in $\mathcal{I}_{na}$ are non-adjacent. It does not contain punctual intervals other than $[0,0]$ as every punctual interval is adjacent to itself.  For example, 
the set $\{[1,2), (2,3], [5,6)\}$ is not a non-adjacent set, while 
$\{[0,0], [0,1), (3,4], [5,6)\}$ is. 
Let $\mathcal{I}_{na}$ denote a set of non-adjacent intervals with end points in $\mathbb{Z} \cup \{-\infty,\infty\}$.
 See full version for an example specification using this logic.
The freeze depth of a $\tptl$ formula $\varphi$, $\fd(\varphi)$ is defined inductively. For a ptopositional formula $prop$, $\fd (prop){=}0$. Also, $\fd(x.\varphi){=}\fd(\varphi)+1$,  and $\fd (\varphi_1 \until \varphi_2){=}\fd(\varphi_1 \since \varphi_2) =\fd(\varphi_1 \wedge \varphi_2){=}\fd(\varphi_1 \vee \varphi_2){=}\mathsf{Max(\fd(\varphi_1),\fd (\varphi_2))}, \fd (\sbf(\varphi)){=}\fd(\sbm(\varphi)){=}\fd(\varphi)$. 

\begin{theorem}
Non-Adjacent 1-$\tptl[\until,\since]$ is more expressive than $\mitl[\until, \since]$. It can also
express the Counting and the Pnueli modalities of \cite{rabin}\cite{rabinovichY}.
\label{thm:me}
\end{theorem}
The straightforward translation of MITL into TPTL in fact gives rise to non-adjacent 1-$TPTL$ formula. Let $\widehat{I}$ abbreviate $T-x \in I$. E.g. MITL formula $a \until_{[2,3]} (b \until_{[3,4]} c)$ translates to $x.(a \until (\widehat{[2,3]} \land x.(b \until (\widehat{[3,4]} \land c)))$. It has been previously shown that $\fut [x.(a \wedge \fut(b \wedge \widehat{(1,2)} \wedge \fut(c \wedge \widehat{(1,2)})))]$, which is in fact a formula of non-adjacent 1-$TPTL$, is inexpressible in $\mtl[\until,\since]$\cite{simoni}. The Pnueli modality $\mathsf{Pn}_I(\phi_1, \ldots, \phi_k)$ expresses that there exist positions $i_1 \leq \ldots \leq i_k$ within (relative) interval $I$ where each $i_j$ satisfies $\phi_j$.  This is equivalent to the non-adjacent 1-$TPTL$ formula
$x. (\fut(\hat{I} \land \phi_1 \land \fut (\hat{I} \land \phi_2 \land \fut( \ldots ))))$. Similarly the (simpler) counting modality can also be expressed.
\\\noindent \textbf{Pnueli $\mathsf{EMTL}$}: There have been several attempts to extend logic $\mtl[\until]$ with regular expression/automaton
modalities \cite{Wilke,KKP17,F18,H19}. We use a generalization of these existing modalities to give the logic $\pnregmtl$. For any finite automaton $A$, let $L(A)$ denote the language of $A$.

Given a finite alphabet $\Sigma$, formulae of $\pnregmtl$ have the following syntax: \\
$\varphi::{=}a~|\varphi \wedge \varphi~|~\neg \varphi~| \fregm^k_{I_1,\ldots,I_k} (\re_1,\ldots, \re_{k+1})(S)~|~ 
\sregm^k_{I_1,\ldots,I_k} (\re_1,\ldots,\re_{k+1})(S)$ where $a \in \Sigma$, 
$I_1, I_2, \ldots I_k \in \intintervaln$ and $\re_1, \ldots \re_{k+1}$ are automata over $2^S$ where $S$ is a set of formulae from $\pnregmtl$. $\fregm^k$ and $\sregm^k$ are the new modalities called future and past {\bf Pnueli Automata} Modalities, respectively, where $k$ is the arity of these modalities.

Let $\rho {=} (a_1, \tau_1),\ldots (a_n, \tau_n) \in T\Sigma^*$, $x,y \in dom(\rho)$, $x{\le} y$ and $S {=} \{\varphi_1,\ldots, \varphi_n\}$ be a given set of $\pnregmtl$ formulae. Let $S_i$ be the exact subset of formulae from $S$ evaluating to true 
at $\rho, i$, and let  $\mathsf{Seg^+}({\rho},{x},{y},S)$ and $\mathsf{Seg^{-}}({\rho},{y},{x},S)$
be the  untimed words $S_x S_{x+1} \ldots  S_y$ and $S_y  S_{y{-}1} \ldots  S_x$ respectively.   
Then, the satisfaction relation for $\rho,i_0$ satisfying a $\pnregmtl$ formula $\varphi$ is defined recursively as follows:\\
$\bullet$ $\rho,i_0{\models}\fregm^k_{I_1,\ldots,I_k}(\re_1,\ldots,\re_{k+1})(S)$ iff 
 ${\exists} {i_0{ {\le} }i_1{\le} i_2 \ldots {\le} i_k {\le} n}$ s.t.\\ $\bigwedge \limits_{w{=}1}^{k}{[(\tau_{i_w} {-} \tau_{i_0} \in I_w)}
 \wedge \mathsf{Seg^+}(\rho, i_{w{-}1}, i_w,S) \in L({\re_w})] \wedge 
 \mathsf{Seg^+}(\rho, i_{k}, n,S) \in L({\re_{k+1}}) $\\
$\bullet$ $\rho,i_0 \models \sregm^k_{I_1,I_2,\ldots,I_k} (\re_1,\ldots,\re_k, \re_{k+1})(S)$ iff  
${\exists} i_0 {\ge} i_1 {\ge} i_2 \ldots {\ge} i_k {\ge} n$ s.t. \\
$\bigwedge \limits_{w{=}1}^{k}[(\tau_{i_0} {-} \tau_{i_w} \in I_w)
 \wedge \mathsf{Seg^{-}}(\rho, i_{w{-}1}, i_w,S) \in L({\re_{w}})] \wedge 
 \mathsf{Seg^{-}}(\rho, i_{k}, n,S) \in L({\re_{k+1}})$. 
\smallskip 
Language of any $\pnregmtl$ formula $\varphi$, as $L(\varphi) = \{\rho | \rho,1 \models \varphi\}$. The Pointed Language of $\varphi$ is defined as $L_{pt}(\varphi) = \{\rho,i | \rho, i \models \varphi\}$.
Given a $\pnregmtl$ formula $\varphi$, its arity is the maximum number of intervals appearing in any $\fregm, \sregm$ modality of $\varphi$. For example, the arity of $\varphi=\fregm^2_{I_1,I_2}(\re_1,\re_2, \re_3)(S_1) \wedge \sregm^1_{I_1} (\re_1,\re_2)(S_2)$ for some sets of formulae $S_1, S_2$ is 2. For the sake of brevity,  $\fregm^k_{I_1,\ldots,I_k}(\re_1,\ldots,\re_{k})(S)$ denotes $\fregm^k_{I_1,\ldots,I_k}(\re_1,\ldots,\re_{k}, \re_{k+1})(S)$ where automata $\re_{k+1}$ accepts all the strings over $S$. 
We define {\bf non-adjacent $\pnregmtl$}, as a subclass where every modality $\fregkm$ and $\sregkm$ is such that $\{\mathsf{I_1, \ldots, I_k}\}$ is a non-adjacent set of intervals.

$\emitl$ of \cite{Wilke} (and variants of it studied in \cite{KKP17}\cite{Kkp18} \cite{F18}\cite{H19}) are special cases of the non-adjacent $\pnregmtl$ modality where the arity is restricted to 1 and the second automata in the argument accepts all the strings. Hence, automaton modality of \cite{Wilke} is of the form $\freg(A)(S)$. 
Let $\emitlinf$ denote the logic $\emitl$ extended with $\fregm$ and $\sregm$ modality where the timing intervals are restricted to be of the form $\langle l, \infty)$ or $\langle 0 , u \rangle$.

We conclude this section defining size of a temporal logic formula.
\\\noindent  \textbf{Size of Formulae}. Size of a formula $\varphi$ denoted by $|\varphi|$ is a measure of how many bits are required to store it. The size of a $\tptl$ formula is defined as the sum of the number of $\until$, $\since$ and Boolean operators and freeze quantifiers in it. For $\pnregmtl$ formulae, $|op|$ is defined as the number of Boolean operators and variables used in it. $|(\fregk_{I_1,\ldots,I_k}(\re_1, \ldots, \re_{k+1})(S)|{=}\sum \limits_{\varphi \in S}(|\varphi|)+ |\re_1|{+}\ldots{+}|\re_{k+1}|$ where $|\re|$ denotes the size of the automaton $\re$ given by sum of number of its states and transitions.

\section{Anchored Interval Word Abstraction}
\label{sec:interval word}
All the logics considered here have the feature that a sub-formula asserts timing constraints on various positions relative to an anchor position; e.g. the position of freezing the clock in TPTL. Such constraints can be symbolically represented as an interval word with a unique anchor position and all other positions carry a set of time intervals constraining the time stamp of the position relative to the time stamp of the anchor. See interval word $\kappa$ in Figure  \ref{fig:norm_example}.We now define these interval words formally.
Let $I_{\nu}\subseteq \intinterval$. 
An $I_{\nu}$-interval word over $\Sigma$ is a word $\kappa$ of the form 
$a_1 a_2  \dots a_n \in 
(2^{\Sigma \cup \{\anch\} \cup  I_\nu})^*$. 
There is a unique $i \in dom(\kappa)$ called the \emph{anchor} 
of $\kappa$ and denoted by $\anch(\kappa)$. At the anchor position $i$,   $a_i \subseteq \Sigma \cup \{\anch\}$, 
and $\anch \in a_i$.  
Let $J$ be any interval in $\I_\nu$. We say that a point $i \in dom(\kappa)$ is a $J$-time restricted point if and only if, $J \in a_i$. 
$i$ is called time restricted point if and only if either $i$ is $J$-time restricted for some interval $J$ in $I_\nu$ or $\anch \in a_i$. 
\smallskip

\noindent \textbf{From $I_\nu$-interval word to Timed Words} : 
Given a $I_\nu$-interval word $\kappa=a_1 \dots a_n$ over $\Sigma$ and 
a timed word $\rho=(b_1, \tau_1)\dots (b_m, \tau_m)$, 
the pointed timed word $\rho, i$ is consistent with $\kappa$ iff 
$dom(\rho){=}dom(\kappa)$, $i{=}\anch(\kappa)$, for all $j\in dom(\kappa)$, $b_j=a_j\cap \Sigma$ and for $j \ne i$, $I \in a_j \cap I_\nu$ implies $\tau_j - \tau_i \in I$. 
Thus, $\kappa$ and $\rho,i$ agree on propositions at all positions, and the time stamp of a non-anchor position $j$ in $\rho$ satisfies every interval constraint in $a_j$ relative to $\tau_i$, the time stamp of anchor position.
$\mathsf{Time(\kappa)}$ denotes the set of all the pointed timed words consistent with a given interval word $\kappa$, and  $\mathsf{Time(\Omega)}{=}\bigcup \limits_{\kappa \in \Omega} (\mathsf{Time(\kappa)})$ for a set of interval words $\Omega$. Note that the ``consistency relation'' is a many-to-many relation. 
\\\noindent {\it Example.}
Let 
$\kappa{=}{\{a,b, (-1,0)\} \{b, (-1,0)\} 
\{a, \anch\} \{b,[2, 3]\}}$ be an interval word over the set of intervals 
$\{(-1,0),[2,3]\}$. 
Consider timed words $\rho$ and $\rho'$ s.t.
 $\rho{=}{(\{a,b\}, 0)(\{b\}, .5), (\{a\}, .95) (\{b\}, 3)}$,$\rho'{=}{(\{a,b\},0)(\{b\},0.8)(\{a\},0.9)(\{b\},2.9)}$. 
\\Then $\rho,3$ as well as $\rho',3$ are  consistent with $\kappa$ while $\rho,2$ is not.  Likewise, for the timed word $\rho'' {=} (\{a,b\}, 0), (\{b\}, 0.5), (\{a\}, 1.1) (\{b\}, 3)$,  $\rho'', 3$ is not consistent with $\kappa$ as $\tau_1-\tau_3 \notin (-1,0)$, as also $\tau_4-\tau_3 \notin [2,3]$.

Let $I_\nu, I_\nu' \subseteq \intinterval$. Let $\kappa=a_1\dots a_n$ and $\kappa'=b_1 \dots b_m$ be $I_\nu$ and $I_\nu'$-interval words, respectively. 
 $\kappa$ is \emph{similar} to $\kappa'$, denoted by $\kappa \sim \kappa'$ if and only if, \\(i) $dom(\kappa){=}dom(\kappa')$, (ii) for all $i \in dom(\kappa)$, $a_i \cap  \Sigma{=}b_i\cap \Sigma$, and \\(iii)$\anch(\kappa)=\anch(\kappa')$. Additionally,
 $\kappa$ is \emph{congruent} to $\kappa'$, denoted by $\kappa \cong \kappa'$, iff $\mathsf{Time}(\kappa){=}\mathsf{Time}(\kappa')$. I.e., $\kappa$ and $\kappa'$ abstract the same set of pointed timed words. 
\\\noindent \textbf{Collapsed Interval Words}. The set of interval constraints at a position can be collapsed into a single interval by taking the intersection of all the 
intervals at that position giving a Collapsed Interval Word.Given an $I_{\nu}$-interval word $\kappa{=}a_1 \dots a_n$, let $\Ii_j = a_j \cap I_{\nu}$.
Let $\kappa'{=}\col(\kappa)$ be the word obtained by replacing $\Ii_j$ with $\bigcap_{I \in \Ii_j} I$ in $a_j$, for all $j{\in}dom(\kappa)$. Note that $\kappa'$ is an interval word over $\Clcap(I_\nu){=}\{I | I{=}\bigcap I', I' \subseteq I_\nu\}$.
 Note that if  for any j, the set $\Ii_j$ contains two disjoint intervals (like $[1,2]$ and $[3,4]$) then $\col(\kappa)$ is undefined.
It is clear that  $\mathsf{Time}(\kappa){=}\mathsf{Time}(\kappa')$.  An interval word $\kappa$ is called \emph{collapsed} iff $\kappa{=}\col(\kappa)$. 
\begin{figure}[t]
    \centering
  \scalebox{0.6}{  \includegraphics{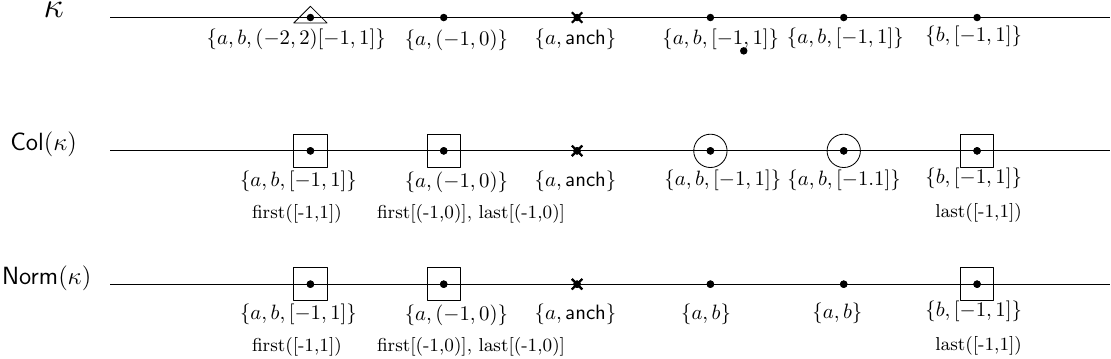}}
    \caption{Point within the triangle has more than one interval. The encircled points are intermediate points and carry redundant information. The required timing constraint is encoded by first and last time restricted points of all the intervals (within boxes). }
    \label{fig:norm_example}
\end{figure}
\\\noindent \textbf{Normalization of Interval Words} An interval $I$ may repeat many times in a collapsed interval word $\kappa$. Some of these occurrences are redundant and we can only keep the first and last occurrence of the interval in the  normalized form of $\kappa$. See Figure  \ref{fig:norm_example}.
For a collapsed interval word $\kappa$ and any $I \in I_\nu$, let $\first(\kappa, I)$ and $\last(\kappa,I)$  denote the positions of first and last occurrence of $I$ in $\kappa$. If $\kappa$ does not contain any occurrence of $I$,  then both $\first(\kappa, I){=}\last(\kappa,I){=}\bot$. We define, $\boundaryint(\kappa){=}\{i | i{\in}dom(\kappa){\wedge}\exists I{\in}I_\nu$ s.t. $(i{=}\first(\kappa,I) {\vee}i{=}\last (\kappa,I){\vee}i{=}\anch(\kappa)) \}$  

The normalized interval word corresponding to $\kappa$, denoted $\normalize(\kappa)$, is defined as 
$\kappa_{nor}=b_1 \dots b_m, $ such that (i) $\kappa_{nor} \sim \col(\kappa)$, (ii) for all $I \in \Clcap(I_\nu)$, $\first(\kappa, I){=}\first(\kappa_{nor}, I)$,  $\last(\kappa, I){=}\last (\kappa_{nor}, I)$, and for all points $j \in dom(\kappa_{nor})$ with $\first(\kappa, I) < j < \last(\kappa,I)$,
$j$ is not a $I$-time constrained point. See Figure  \ref{fig:norm_example}.
Hence, a normalized word is a collapsed word where for any $J\in \Clcap(I_\nu)$ there are at most two $J$-time restricted points.  This is the key property which will be used to reduce 1-$\tptl$ to a finite length $\pnregmtl$ formulae.

\begin{lemma}
\label{lem:normalize}
$\kappa \cong \normalize(\kappa)$.  Note, $\normalize(\kappa)$ has at most $2{\times}|I_\nu|^2{+}1$ restricted points.
\end{lemma}
The proof follows from the fact that $\kappa \cong \col(\kappa)$ and since $ \col(\kappa) \sim \normalize(\kappa)$, the set of timed words consistent with any of them will have identical untimed behaviour.
     For the timed part, the key observation is as follows. For some interval $I \in I_{\nu}$, let $i'{=}\first(\kappa,I), j'{=}\last(\kappa, I)$. Then  for any $\rho,i$ in $\mathsf{Time}(\kappa)$, points $i'$  and $j'$ are within the interval $I$ from point $i$. Hence, any point $i' \le i'' \le j'$ is also within interval $I$ from $i$. Thus, the interval  $I$ need not be explicitly mentioned at intermediate points. 
  The full proof can be found in the full version. 


\section{1-$\tptl$  to $\pnregmtl$}
\label{sec:tptltopnregmtl}
In this section, we reduce a 1-$\tptl$ formula into an equisatisfiable $\pnregmtl$ formula. First, we consider 1-$\tptl$  formula with a single outermost freeze quantifier (call these simple $\tptl$ formulae) and give the reduction. More complex formulae can be handled by applying the same reduction recursively as shown in the first step. 
For any set of formulae $S$, let $\bigvee S$ denote $\bigvee \limits_{s\in S}s$.
\\A $\tptl$ formula is said to be \emph{simple} if it is of the form $x.\varphi$ where, $\varphi$ is a 1-$\tptl$  formula with no freeze quantifiers. Let $\I_\nu \subseteq \intinterval$. Let $\psi{=}x.\varphi$  be a simple $\I_\nu$-$\tptl$  formula and let $\Clcap(\I_\nu){=}I_\nu$. 
We construct a $\pnregmtl$ formula $\phi$, such that $\rho,i \models \psi{\iff}\rho, i \models \phi$. We break this construction into the following steps:
\\1) We construct an $\ltl$ formula $\alpha$ s.t. $L(\alpha)$ contains only $I_\nu$-interval words and $\rho,i \models \psi$ iff $\rho, i \in \mathsf{Time}(L(\alpha))$. Let $A$ be the NFA s.t. $L(A)=L(\alpha)$.
Let $W$ be the set of all  $I_\nu$-interval words. 
\\2) We partition $W$ into finitely many types, each \emph{type}, capturing a certain relative ordering between first and last occurrences of intervals from $I_\nu$ as well as $\anch$.  
Let $\mathcal{T}(I_\nu)$ be the finite set of all types.  
\\3) For each type $\seq \in \mathcal{T}$, we construct an NFA $A_{\seq}$ such that $L(A_\seq)=\normalize(L(A) \cap W_\seq)$, where $W_\seq$ is the set of all the $I_\nu$-interval words of type $\seq$. 
\\4)For every type $\seq$, using the $A_{\seq}$ above, we construct a $\pnregmtl$ formula $\phi_\seq$ such that, $\rho, i \models 
\phi_\seq$ if and only if $\rho, i \in \mathsf{Time}(L(A_\seq))$. 
The desired $\phi = \bigvee \limits_{\seq \in \mathcal{T}(I_\nu)} ~\phi_\seq$.
 Hence, $L_{pt}(\phi){=}\bigcup  \limits_{\seq \in \mathcal{T}} \mathsf{Time}(L(A_\seq))
       {=}\mathsf{Time}(L(A)) = L_{pt}(\psi)$.
\\\noindent\textbf{1a) Simple $\tptl$  to $\ltl$ over Interval Words}:
As above, $\psi{=}x.\varphi$. Consider an $\ltl$ formula $\alpha{=}\fut[\ltl(\varphi) \wedge \anch \wedge \neg (\fut (\anch) \vee \past (\anch))] \wedge \sbf (\bigvee \Sigma)$ over $\Sigma'{=}\Sigma \cup \I_\nu \cup \{\anch\}$, where $\ltl(\varphi)$ is the the LTL formula obtained from $\varphi$ by replacing clock constraints $T-x \in I$ with $I$ and $x-T \in I$ with $-I$.  Then all words in $L(\alpha)$ are $\I_\nu$-interval words.
\begin{theorem}
\label{thm:tptl-ltl}
For any timed word $\rho$, $i\in dom(\rho)$, and any clock valuation $v$, $\rho, i, v \models \psi{\iff}\rho, i \in \mathsf{Time}(L(\alpha))$. 
\end{theorem}
{\it Proof Sketch}. Note that for any timed word $\rho$ and $i \in dom(\rho)$, $\rho,i, [x \leftarrow\tau_i] \models \varphi $ is equivalent to $\rho, i \models \psi$ since $\psi=x.\varphi$. 
Let $\kappa$ be any $\I_\nu$-interval word over $\Sigma$ with $\anch(\kappa){=}i$. 
It can be seen that  if $\kappa, i \models \ltl(\varphi)$ then for all 
$\rho,i \in \mathsf{Time}({\kappa})$ we have $\rho,i \models \psi$.
Likewise, if $\rho,i \models \psi$ for a timed word $\rho$, then there exists some $\I_\nu$-interval word over $\Sigma$ such that $\rho, i \in \mathsf{Time}(\kappa)$ and $\kappa, i \models \ltl(\varphi)$. 

Illustrated on an example, if $\psi=x.\varphi$ and $\varphi{=}\fut (x \in I \wedge a)$. 
Then $\rho,i \models \varphi$ iff there exists a point $j$ within an interval $I$ from $i$, where $a$ holds. 
Now consider $\alpha{=}\fut^{w}[(I \wedge a)  \wedge \anch \wedge \neg (\fut (\anch) \vee \past (\anch))] \wedge \wB (\bigvee \Sigma)$ whose language consists of interval words $\kappa$ such that there is a point ahead of the anchor point $i$ where both $a$ and $I$ holds. 
Clearly, words in $\mathsf{Time}({\kappa})$ are such that they contain a point $j>i$ within an interval $I$ from point $i$ where $a$ holds. 
Hence, $\rho, i \models \psi$ if and only if $\rho, i \in \mathsf{Time}(\{\kappa ~\mid~ \kappa, i \models \ltl(\varphi)\})$. Moreover, $\kappa \in L(\alpha)$ if and only if $\kappa, i \models \ltl(\varphi)$ and $\anch(\kappa){=}i$. The 
full proof is in the full version. 
\smallskip 

\noindent\textbf{1b) $\ltl$ to NFA over Collapsed Interval Words}.
It is known that for any $\ltl[\until,\since]$ formula, one can construct an equivalent NFA with at most exponential number of states \cite{gastin-oddoux}. We reduce the LTL formula $\alpha$ to an equivalent NFA $A_\alpha{=}(Q, \init, 2^{\Sigma'}, \delta', F)$ over $I_{\nu}$-interval words, where $\Sigma'=2^{\Sigma \cup I_{\nu} \cup \{\anch\}}$.  
From $A_{\alpha}$, we  construct an automaton $A{=}(Q, \init, 2^{\Sigma'}, \delta, F)$  s.t. $L(A){=}\col(L(A_\alpha))$. 
Automaton  $A$ is obtained from $A_\alpha$ by replacing  the set of intervals $\I$ on the transitions by the single interval $\bigcap \I$. In case $\exists I_1, I_2 \in \I$ s.t. $ I_1 \cap I_2 = \emptyset$ (i.e. with contradictory interval constraints), the transition is omitted in $A$.
This gives $L(A){=}\col(L(A_\alpha))$. 
\\
\smallskip 
\noindent\textbf{2) Partitioning Interval Words}.
We discuss here how to partition $W$, the set of all collapsed $I_\nu$-interval words, into finitely many classes. Each class is characterized by its {\bf type} given as a  finite sequences $\seq$ over $I_{\nu} \cup \{\anch\}$.
For any collapsed $w \in W$, its type $\seq$  gives an ordering between $\anch (w)$, $\first(w,I)$ and $\last(w,I)$ for all $I \in I_\nu$, such that, any $I \in I_{\nu}$ appears at most twice 
and $\anch$ appears exactly once in $\seq$. For instance, $\seq=I_1 I_1 \anch  I_2 I_2$ is a sequence 
different from $\seq'=I_1I_2 \anch I_2I_1$ since the relative orderings between the first and last occurrences of $I_1, I_2$ and $\anch$ differ in both. Let the set of types $\mathcal{T}(I_\nu)$ be the set of all such sequences;  by definition, $\mathcal{T}(I_\nu)$ is finite. 
Given $w \in W$, let $\boundaryint(w){=}\{i_1, i_2, \ldots, i_k\}$ be the 
positions of $w$ which are either $\first(w,I)$ or $\last(w,I)$ for some $I \in I_{\nu}$ 
or is $\anch(w)$. Let $w\downarrow_{\boundaryint(w)}$ be the subword of $w$ obtained by projecting $w$ to the positions in $\boundaryint(w)$, restricted to the sub alphabet $2^{I_{\nu}}\cup\{\anch\}$. For example, \\$w=\{a,I_1\}\{b,I_1\}\{c,I_2\}\{\anch,a\}\{b,I_1\}\{b,I_2\}\{c,I_2\}$
gives $w\downarrow_{\boundaryint(w)}$ as \\ $I_1I_2\anch I_1I_2$. 
Then $w$ is in the partition $W_{\seq}$ iff  $w\downarrow_{\boundaryint(w)}=\seq$. Clearly, $W=\bigcup_{\seq \in \mathcal{T}(I_\nu)}W_{\seq}$.
Continuing with the example above, $w$ is a collapsed $\{I_1, I_2\}$-interval word over $\{a,b,c\}$, with $\boundaryint(w){=}\{1,3,4,5,7\}$, and $w{\in} W_{\seq}$ for $\seq=I_1I_2\anch I_1 I_2$, while 
$w\notin W_{\seq'}$ for $\seq'=I_1I_1\anch I_2 I_2$. 
\smallskip 

\noindent \textbf{3) Construction of NFA for each type}:
Let $\seq$ be any sequence in $\mathcal{T}(I_\nu)$. In this section, given $A{=}(Q, \init, 2^{\Sigma'}, \delta, F)$ as constructed above, we construct an NFA 
$A_\seq{=}(Q \times \{1,2,\ldots |\seq|+1\}\cup \{\bot\}, (\init, 1), 2^{\Sigma'}, \delta_\seq, F \times \{|\seq|+1\})$ such that $L(A_\seq){=}\normalize(L(A) \cap W_\seq)$. Thus, $\bigcup_{\seq\in \mathcal{T}(I_\nu)}L(A_\seq) = \normalize(L(A))$. Thus, $\bigcup_{\seq\in \mathcal{T}(I_\nu)} \mathsf{Time}(L(A_\seq)) = \mathsf{Time}(\normalize(L(A))) = \mathsf{Time}(L(A)) = L(\psi)$.
 Intuitively, the second element of the state makes sure that only normalized words of type $\seq$ are accepted.
From $(q,j)$, $A_\seq$ is allowed to read a set $S \subseteq \Sigma$ (containing no time interval or $\anch$ and hence an unrestricted point) or it can read a set $S\cup \{I\}$ where $S \subseteq \Sigma$ and $J = \seq[j]$ (containing time interval/anchor $\seq[j]$). In case of latter, the $A_\seq$ ends up with a state of the form $(q',j+1)$ 
if and only if there is a transition in $A$ of the form $q \stackrel{S\cup{J}}{\rightarrow}q'$. 
In case of the former, it non-determinstically proceeds to a state $(q',j)$ if and only if, in automaton $A$, there is a transition of the form $q \stackrel{S}{\rightarrow}q'$ or 
$q \stackrel{S\cup{J}}{\rightarrow}q'$ where $\first(J,w)$ has already been read and $\last (J,w)$ is yet to be read in the future (that is, $\exists j'' j', j'< j \le j'' \wedge \seq[j'] = \seq[j''] = J$).
The detailed construction as well as the proof for Lemma \ref{lem:nfatonfaseq}
can be found in the full version. 
Let $W_\seq$ denote set of $I_\nu$-interval words of type $\seq$.
 \begin{lemma}
$L(A_\seq){=}\normalize(L(A) \cap W_\seq)$. Hence, $\bigcup \limits_{\seq \in \mathcal{T}(I_\nu)} L(A_\seq) {=}\normalize(L(A))$.
\label{lem:nfatonfaseq}
\end{lemma}

\noindent Our next step is to reduce the NFAs $A_{\seq}$ corresponding to each type $\seq$ to $\pnregmtl$. 
The  words in $L(A_{\seq})$ are all normalized, and have at most $2|I_\nu|+1$-time restricted points. Thanks to this, its corresponding timed language can be expressed using $\pnregmtl$ formulae with arity at most $2|I_\nu|$.
\smallskip 

\noindent\textbf{4) Reducing NFA of each type to $\pnregmtl$}:
Next, for each $A_\seq$ we construct $\pnregmtl$ formula $\phi_\seq$  such that, for a timed word $\rho$ with $i \in dom(\rho), \rho, i \models \phi_\seq$ iff $\rho, i {\in} \mathsf{Time}(L(A_\seq))$. For any NFA $N = (St,\Sigma, i,Fin,\Delta)$, $q \in Q$ 
$F' \subseteq Q$,  let $N[q,F'] = (St,\Sigma, q, F', \Delta)$.
For brevity, we denote $N[q,\{q'\}]$ as $N[q,q']$.
We denote by $\rev(N)$, the NFA $N'$ that accepts the reverse of $L(N)$. 
The right/left concatenation of $a{\in}\Sigma$ with $L(N)$ is denoted  $N{\cdot}a$ and  $a{\cdot}N$ respectively.  
 \begin{lemma}
\label{lem:nfaseqtopnregmtl}
We can construct a $\pnregmtl$ formula $\phi_\seq$ with arity $\le$ $|I_\nu|^2$ and size $\mathcal{O}(|A_\seq|^{|\seq|})$ containing intervals from $I_\nu$ s.t. $\rho, i \models \phi_\seq$ iff $\rho, i \in \mathsf{Time}(L(A_\seq))$. 
\end{lemma}
\begin{proof}
\begin{figure}
    \centering\scalebox{1.1}{
    \includegraphics{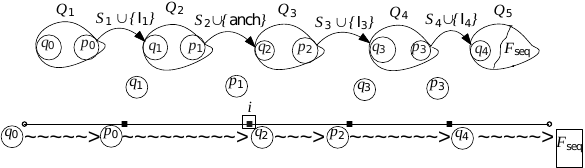}}
    \caption{Figure representing set of runs $A_{\mathsf{I_1 \anch I_3 I_4}}$ of type $Qseq$  where each $S_i \subseteq \Sigma$ and each sub-automaton $Q_i$ has only transitions
    without any intervals. Here $Qseq = T_1 T_2 T_3 T_4$, for $1\le i \le 4$, $T_i = (p_{i-1} \stackrel{S_i\cup \{I_i\}}{\rightarrow} q_i),$ $I_2 = \{\anch\}$.}
    \label{fig:nfatopnregmtl}
\end{figure}
Let $\seq{=}I_1\ I_2\ \ldots\ I_n$, and  $I_j{=}\anch$ for some $1{\le}j {\le}n$. 
 Let $\Gamma=2^{\Sigma}$ and
 \\$\Qseq{=}T_1\ T_2\ \ldots T_n$ be a sequence of transitions of $A_{\seq}$ where for any $1 \le i \le n$,  $T_i{=}p_{i-1} \stackrel{S'_{i}}{\rightarrow} q_{i}$, $S'_i=S_i \cup \{I_i\}$, $S_i \subseteq \Sigma$, 
$p_{i-1} \in Q \times \{i-1\}$, $q_{i} \in Q \times \{i\}$. Let $q_0{=}(\init, 1)$.
We define $\mathsf{R}_{\Qseq}$ as set of accepting runs containing transitions $T_1\ T_2\ \ldots T_n$. 
Hence the runs in $\mathsf{R}_{\Qseq}$ are of the following form:
\\
{$T_{0,1}~T_{0,2} \ldots T_{0,m_0}~T_1~~T_{1,1}~ \ldots T_{1,m_1}~T_{2}~~\cdots\cdots~~T_{n{-}1,1}~T_{n{-}1,2} \ldots T_{n}~~T_{n,1} \ldots  T_{n{+}1}\\$}
where the source of the transition $T_{0,1}$ is $q_0$ and the target of the transition $T_{n+1}$ is any accepting state of $A_\seq$. Moreover, all the transitions $T_{i,j}$ for $0\le i\le n$, $1\le j \le n_i$ are of the form $(p' \stackrel{S_{i,j}}{\rightarrow} q')$ where $S_{i,j} \subseteq \Sigma$ and $p',q' \in Q_{i+1}$. Hence, only $T_1, T_2, \ldots T_n$ are labelled by any interval from $I_\nu$. Moreover, only on these transitions the counter (second element of the state) increments. Let $\re_i = (Q_{i}, 2^{\Sigma}, q_{i-1}, \{p_{i-1}\}, \delta_\seq) \equiv A_{\seq}[q_{i-1},p_{i-1}]$ for $1 \le i \le n$ and \\$\re_{n+1}{=} (Q_{n+1}, 2^{\Sigma}, q_{n}, F_{\seq},\delta_{\seq}){\equiv}A[q_n,F]$. 
Let $\mathcal{W}_{Qseq}$ be set of words associated with any run in $\mathsf{R}_{Qseq}$.  In other words, any word $w$ in $\mathcal{W}_{Qseq}$ admits an accepting run on $A$ which starts from $q_0$ reads letters without intervals (i.e., symbols of the form $S \subseteq \Sigma$)ends up at $p_0$, reads $S'_1$, ends up at $q_1$ reads letters without intervals, ends up and $p_1$, reads $S'_2$ and so on.
Refer figure \ref{fig:nfatopnregmtl} for illustration.
Hence, $w \in W_{Qseq}$ if and only if $w \in L(\re_1).S'_1. L(\re_2).S'_2. \cdots. L(\re_n).S'_n.L(\re_{n+1})$.
\\ 
Let $\re'_{k}{=}S_{k-1}\cdot{\re_{k}}\cdot S_k$ for $1{\le}k{\le}n{+}1$, with $S_0{=}S_{n+1}{=}\epsilon$ \footnote{ We $\re'_k$ instead of $\re_k$ in the formulae below due to the strict inequalities in the semantics of $\pnregmtl$ modalities}.   
Let $\rho{=}(b_1, \tau_1) \ldots (b_m, \tau_m)$ be a timed word over $\Gamma$. Then
\\$\rho,i_j \in \mathsf{Time}(W_{Qseq})$ iff $\exists$  $0{\le} i_1 {\le} i_2 {\le} \ldots {\le} i_{j-1} {\le} i_j {\le} i_{j+1} {\le} \ldots {\le} i_n {\le} m$ s.t. 
\\$\bigwedge\limits_{k =1}^{j-1}[(\tau_{i_k}{-}\tau_{i_j} \in I_k)
  \wedge \mathsf{Seg^-}(\rho, i_{k+1}, i_{k},\Gamma) \in L(\rev({\re'_{k}}))] \wedge \bigwedge \limits_{k{=}j}^{n}[(\tau_{i_k}{-}\tau_{i_j} \in I_k)
    \wedge \mathsf{Seg^+}(\rho, i_{k}, i_{k+1},\Gamma) \in L(\re'_k)]$, where $i_0{=}0$ and $i_{n+1}{=}m$.
Hence, by semantics of $\fregk$ and $\sregk$ modalities, $\rho,i \in \mathsf{Time}(\mathcal{W}_{Qseq})$ if and only if $\rho, i{\models} \phi_{\qseq}$ where  \\
$\phi_{\qseq}=\sregm^{j}_{I_{j-1},\ldots,I_{1}} (\rev(\re'_1),\ldots,\rev(\re'_j))(\Gamma) \wedge \fregm^{n-j}_{I_{j+1},\ldots,I_{n}} (\re'_{j+1},\ldots,\re'_{n+1})(\Gamma)$. \\Let $\mathsf{State{-}seq}$ be set of all possible sequences of the form $\Qseq$.
As $A_\seq$ accepts only words which has exactly $n$ time restricted points, the number of possible sequences of the form $\Qseq$ is bounded by $|Q|^{n}$. Hence any word $\rho, i \in \mathsf{Time}(L(A_\seq))$ iff $\rho, i \models
\phi_{\mathsf{seq}}$ where $\phi_{\mathsf{seq}} =     
\bigvee \limits_{\mathsf{qseq} \in \mathsf{State{-}seq}} \phi_{\mathsf{qseq}}$. 
Disjuncting over all possible sequences $\seq{\in}\mathcal{T}(I_\nu)$ we get formula $\phi$ and the following lemma.
\end{proof}
\begin{lemma}
\label{lem:nfatopnregmtl}
    Let $L(A)$ be the language of $I_\nu$-interval words definable by a NFA $A$. 
    We can construct a $\pnregmtl$ formula $\phi$  s.t. $\rho,i \models \phi$ iff $\rho,i \in \mathsf{Time}(L(A))$.
\end{lemma}
Note that, if $\psi$ is a simple 1-$\tptl$ formula with intervals in $\I_\nu$, then the equivalent $\pnregmtl$ formula $\phi$ constructed above contains only interval in  $\Clcap(\I_\nu)$.  Hence, we have the following theorem.
\begin{theorem}
For a simple non-adjacent 1-$\tptl$ formula $\psi$ containing intervals from $\I_\nu$,  we can construct a non-adjacent $\pnregmtl$ formula  $\phi$, s.t. for any valuation $v$, $\rho,i,v {\models} \psi$ iff  $\rho, i {\models} \phi$ where, $|\phi|{=}O(2^{Poly(|\psi|)})$ and arity of $\phi$ is $O(|\I_\nu|^2)$.  
\label{thm:natptltonapnregmtl}
\end{theorem}
This is a consequence of  Theorem \ref{thm:tptl-ltl}, Lemma \ref{lem:nfatonfaseq} and Lemma \ref{lem:nfatopnregmtl}. A formal proof appears in the full version 
For the complexity : The size $\ltl$ formula $\alpha$ constructed from $\psi$ (in \textbf{1a))}) is linear in $\psi$. The translation from LTL formula $\alpha$ to NFA $A$ has a complexity $\mathcal{O}(2^{|\alpha|}) = \mathcal{O}(2^{|\psi|})$. Let $I_\mu = \Clcap(\I_\nu)$. Hence, $|I_\mu| = \mathcal{O}(|\I_\nu|^2)$. 
The size of $A_{\seq}$ is $\mathcal{O}(|\seq| \times 2^{(|\psi|)}) = \mathcal{O}(2^{Poly(|\psi|)})$ as $|\seq| \le 2 \times |I_\mu| = \mathcal{O}(|\I_\nu|^2) = \mathcal{O}(|\psi|^2)$. Next, 
$|\phi_\seq| = \mathcal{O}(|A_\seq|^{|\seq|}) = \mathcal{O}(2^{Poly(|\psi|)})$. $|T(\I_\nu)| = O(2^{Poly(n)})$. Hence, $|\phi|=O(2^{Poly(n, |Q|)}) = O(2^{Poly(|A|)})$. Moreover, the arity of $\phi$ is also bounded by $2\times |\Clcap(I_\nu)|$. Note that, $|\Clcap(I_\nu)| \le |\I_\nu|^2$. Moreover, $\Clcap(I_\nu)$ is non-adjacent iff $I_\nu$ is. This result is lifted to a (non-simple)  1-$\tptl$ formula $\psi$ as follows: for each occurrence of a subformula $x.\varphi_i$ in $\psi$, introduce a new propositional variable $a_i$ and replace $x.\varphi_i$  with $a_i$. After replacing all such,  we are left with the 
outermost freeze quantifier. Conjunct $\bigwedge\limits_{i=1}^{m} \wB (a_i \leftrightarrow x.\varphi_i)$ to the replaced 
formula obtaining a simple 1-$\tptl$ formula $\psi'$, equisatisfiable to $\psi$.  
Apply the procedure above to each of the $m+1$ conjuncts of $\psi'$ resulting in  $m+1$ equivalent non-adjacent $\pnregmtl$ formulae $\varphi'_i$. The conjunction of $\varphi'_i$ is the non-adjacent $\pnregmtl$ formula equisatisfiable with $\psi$, giving Theorem \ref{thm: tptltopnregmtlsat}. 
\begin{theorem}
\label{thm: tptltopnregmtlsat}
Any non-adjacent 1-$\tptl$ formula $\psi$ with intervals in $\I_\nu$, can be reduced to a non-adjacent $\pnregmtl$, $\phi$, with $|\phi|= 2^{Poly(|\psi|)}$ and arity of $\phi{=}O(|\I_\nu|^2)$ such that $\psi$ is satisfiable if and only if $\phi$ is.
\end{theorem}
\section{Satisfiability Checking for non-adjacent $\pnregmtl$}
\label{sec:pnregsatisfiability}
\begin{theorem}
Satisfiability Checking for non-adjacent $\pnregmtl$ and non-
adjacent 1-$\tptl$ are decidable with EXPSPACE complete complexity.
\label{thm:main}
\end{theorem}
The proof is via a satisfiability preserving reduction to logic $\emitl_{(0,\infty)}$ resulting in a formula whose size is at most exponential in the size of the input non-adjacent $\pnregmtl$ formula.  
Satisfiability checking for $\emitl_{0,\infty}$ is  PSPACE complete \cite{H19}. This along with our construction  
implies an EXPSPACE decision procedure for satisfiability checking of non-adjacent $\pnregmtl$. The EXPSPACE lower bound 
follows from the EXPSPACE hardness of sublogic $\mitl$. The same complexity also applies  to non-adjacent 1-$\tptl$, using the reduction in the previous section.  We now describe the technicalities associated with our reduction. We use the technique of equisatisfiability modulo oversampling \cite{time14}\cite{khushraj-thesis}. Let $\Sigma$ and $\Ovs$ be disjoint set of propositions.
Given any timed word $\rho$ over $\Sigma$, we say that a word $\rho'$ over $\Sigma \cup \Ovs$ is an oversampling of $\rho$ if $|\rho| \le |\rho'|$ and when we delete the symbols in $\Ovs$ from $\rho'$ we get back $\rho$. Intuitively, $\Ovs$ are set of propositions which are used to label oversampling points only. Informally, a formulae $\alpha$ is  equisatisfiable modulo oversampling to formulae $\beta$ if and only if for every timed word  $\rho$ excepted by $\beta$ there there exists an oversampling of $\rho$ accepted by $\alpha$ and, for every timed word $\rho'$ accepted by $\alpha$ its projection is accepted by $\alpha$. Note that when $|\rho'| > |\rho|$, $\rho'$ will have some time points where no proposition from $\Sigma$ is true. These new points are called oversampling points. Moreover, we say that any point $i' \in dom(\rho')$ is an old point of $\rho'$ corresponding to $i$ iff $i'$ is the $i^{th}$ point of $\rho'$ when we remove all the oversampling points. 
For the rest of this section, let $\phi$ be a non-adjacent $\pnregmtl$ formula over $\Sigma$. We break down the construction of an $\emitl_{0,\infty}$ formula $\psi$ as follows. \\
\noindent 1) Add oversampling points at every integer timestamp using  $\varphi_{\ovs}$ below,\\
\noindent 2) Flatten the $\pnregmtl$ modalities to get rid of nested automata modalities, obtaining an  equisatisfiable formula $\phi_{flat}$,\\
\noindent 3) With the help of oversampling points, assert the properties expressed by $\pnregmtl$ subformulae $\phi_i$ 
of $\phi_{flat}$ using only $\emitl$ formulae,\\
\noindent 4) Get rid of bounded intervals with non-zero lower bound, getting the required $\emitl_{0,\infty}$ formula $\psi_i$.  Replace  $\phi_i$ with $\psi_i$ in $\phi_{flat}$ getting $\psi$.
\\Let $\Last{=}\sbf \bot$ and $\Lastts{=}\sbf \bot \vee (\bot \until_{(0,\infty)} \top)$. $\Last$ is true only at the last point of any timed word. Similarly, $\Lastts$, is true at a point $i$ if there is no next point $i+1$ with the same timestamp $\tau_i$.
Let $\cmax$ be the maximum constant used in the intervals appearing in $\phi$.\\
\noindent\textbf{1) Behaviour of Oversampling Points}. We oversample  timed words over $\Sigma$ by adding new points where only propositions from $\Int$ holds, where  $\Int \cap \Sigma = \emptyset$.  
Given a timed word $\rho$ over $\Sigma$, consider an extension of $\rho$ called $\rho'$, by extending the alphabet $\Sigma$ of $\rho$ to 
$\Sigma' = \Sigma \cup \Int$. Compared to $\rho$, $\rho'$  has extra points called \emph{oversampling} points, where 
 $\neg \bigvee \Sigma$ (and $\bigvee \Int$) hold. These extra points are added at all integer timestamps, in such a way that if $\rho$ already  has points with integer time stamps, then the oversampled point with the same time stamp  
 appears last among all points with the same time stamp  in $\rho'$. 
  We will make use of these oversampling points to reduce the $\pnregmtl$ modalities into $\emitl_{0,\infty}$. These oversampling points are labelled with a modulo counter $\Int{=}\{\nt_0,\nt_1,\ldots, \nt_{\cmax-1}\}$. 
The counter is initialized to be $0$ at the first oversampled point with timestamp $0$ and is incremented, modulo $\cmax$, after exactly one time unit till the last point of $\rho$.  
Let $i \oplus j{=}(i+ j) \% \cmax$. 
The  oversampled behaviours are expressed using the formula $\varphi_{\ovs}$:  
$\{\neg \fut_{(0,1)} \bigvee \Int \wedge \fut_{[0,1)} \nt_{0}\} \wedge $ \\
$ \{\bigwedge \limits_{i=0}^{\cmax-1}\wB\{(\nt_i{\wedge}\fut(\bigvee \Sigma))\rightarrow (\neg \fut_{(0,1)} (\bigvee \Int) \wedge \fut_{(0,1]} (\nt_{i \oplus 1} \wedge(\neg \bigvee \Sigma) \wedge \Lastts))\}$. 
to an extension $\rho'$ given by $\ext(\rho){=}\rho'$ iff \textbf{(i)}$\rho$ can be obtained from $\rho'$ by deleting oversampling points and \textbf{(ii)}$\rho' \models \varphi_{\ovs}$. Map  $\ext$ is well defined as for any $\rho$, $\rho'=\ext(\rho)$ if and only if $\rho'$ can be constructed from $\rho$ by appending oversampling points at integer timestamps and labelling $k^{th}$ such oversampling point (appearing at time $k{-}1$) with $\nt_{k{\%}\cmax}$.
\\\noindent \textbf{2) Flattening}. Next, we  flatten $\phi$ to eliminate the nested $\fregkm$ and $\sregkm$ modalities while preserving satisfiability. Flattening is well studied  \cite{deepak08}, \cite{time14}, \cite{khushraj-thesis}, \cite{H19}.
The idea is to associate a fresh witness variable $b_i$ to each subformula  $\phi_i$ which needs to be flattened. 
This is achieved using the \emph{temporal definition} $T_i=\wB((\bigvee \Sigma \wedge \phi_i) \leftrightarrow b_i)$  and replacing 
$\phi_i$ with $b_i$ in $\phi$,  $\phi''_i{=}\phi[b_i /\phi_i]$, where  $\wB$ is the weaker form of 
$\sbf$ asserting at the current point and strict future.   Then, $\phi'_i{=}\phi''_i \wedge T_i \wedge \bigvee\Sigma$ is equisatisfiable to $\phi$. 
Repeating this across all  subformulae of $\phi$, we obtain $\phi_{flat}{=}\phi_t \wedge T $ 
over the alphabet  $\Sigma'{=}\Sigma \cup W$, where $W$ is the set of the witness variables, $T=\bigwedge_i T_i$, $\phi_t$ is a propositional logic formula over $W$.  Each $T_i$ is of the form $\wB(b_i \leftrightarrow (\phi_f \wedge \bigvee\Sigma))$ where $\phi_f{=}\fregnm(\re_1, \ldots, \re_{n+1})(S)$ (or uses $\sregnm$) and $S \subseteq \Sigma'$. 
For example, consider the formula $\phi=\fregm^2_{(0,1)(2,3)}(\mathcal{A}_1, \mathcal{A}_2,\mathcal{A}_3)(\{\phi_1 , \phi_2\})$, where \\$\phi_1 = \sregm^2_{(0,2)(3,4)}(A_4,A_5, A_6)(\Sigma),\phi_2 = \sregm^2_{(1,2)(4,5)}(A_7,A_8, A_9)(\Sigma)$.   
Replacing the $\phi_1, \phi_2$ modality with witness propositions 
$b_1, b_2$, respectively, we get \\$\phi_t=\fregm^2_{(0,1)(2,3)}(A_1,A_2,A_3)(\{b_1, b_2\}) \wedge T$, 
where
\\$T=\wB(b_1 \leftrightarrow (\bigvee \Sigma \wedge  \phi_1)) \wedge \wB(b_2 \leftrightarrow (\bigvee \Sigma \wedge \phi_2))$, $A_1,A_2,A_3$ are automata constructed from $\mathcal{A}_1, \mathcal{A}_2, \mathcal{A}_3$, respectively, by replacing $\phi_1$ by $b_1$ and $\phi_2$ by $b_2$ in the labels of their transitions. Hence, $\phi_{flat}=\phi_t \wedge T$ is obtained by flattening the $\fregkm,\sregkm$ modalities.
\\\noindent \textbf{3) Obtaining  equisatisfiable $\emitl$ formula $\psi_f$  for the $\pnregmtl$ formula $\phi_f$ in 
each $T_i=\wB(b_i \leftrightarrow (\phi_f \wedge \bigvee\Sigma))$}. The next step is to replace all the $\pnregmtl$ formulae occurring in temporal definitions $T_i$. We use oversampling to construct the formula $\psi_f$ : for a timed word $\rho$ 
over $\Sigma$, $i \in dom(\rho)$, there is an extension $\rho'=\ext(\rho)$ over 
an extended alphabet $\Sigma'$, and a point $i' \in dom(\rho')$ which is an old point 
corresponding to $i$ such that 
$\rho', i' \models \psi_f$ iff $\rho, i \models \phi_f$.  
\smallskip 
Consider $\phi_f=\fregnm(\re_1, \ldots, \re_{n+1})(S)$ where $S \subseteq \Sigma'$. 
Wlg, we assume:
\\$\bullet$ \textbf{[Assumption1]}: $\inf(\mathsf{I_1}) \le \inf(\mathsf{I_2}) \le \ldots \le \inf(\mathsf{I_n})$ and $\sup(\mathsf{I_1}) \le \ldots \le \sup(\mathsf{I_n})$. This is wlog, since the check for $\re_{j+1}$ cannot start before the check of $\re_{j}$ in case of $\fregnm$ modality (and vice-versa for $\sregnm$ modality) for any $1\le j \le n$.
\\$\bullet$ \textbf{[Assumption 2]}: Intervals $\mathsf{I_1, \ldots I_{n-1}}$ are bounded intervals. Interval $\mathsf{I_{n}}$ may or may not be bounded. This is also wlog
\footnote{Unbounded intervals can be eliminated using $\fregk_{\mathsf{I_1, I_2, \ldots, I_{k{-}2},} [l_1, \infty) [l_2, \infty)}(\re_1, \ldots, \re_{k+1}) {\equiv}\\ \fregk_{\mathsf{I_1, I_2, \ldots, I_{k{-}2},} [l_1, \cmax) [l_2, \infty)}(\re_1, \ldots, \re_{k+1}) {\vee} \fregm^{k-1}_{\mathsf{I_1, I_2, \ldots, I_{k{-}2},}  [l_2, \infty)}(\re_1, \ldots, \re_{k-1},\re_{k} \cdot \re_{k+1})$.}.
\\\noindent Let $\rho{=}(a_1, \tau_1) \dots (a_n, \tau_n)\in T\Sigma^*$, $i \in dom(\rho)$. Let $\rho'{=}\ext(\rho)$ be defined by 
$(b_1, \tau'_1)\dots (b_m, \tau'_m)$ with $m \geq n$, and each $\tau'_i$ is a either a new integer timestamp 
not among $\{\tau_1, \dots, \tau_n\}$ or is some $\tau_j$. 
Let  $i'$ be an old point in $\rho'$ corresponding to $i$. Let $i'_0{=}i'$ and $i'_{n+1} = |\rho'|$. 
 $\rho, i \models \phi_f$ iff $\mathsf{cond} \equiv \exists i'\le i_1' \le \ldots \le i'_{n+1} \bigwedge \limits_{g{=}1}^n (\tau'_{i'_g} -\tau'_{i'} {\in} \mathsf{I}_g \wedge \rho',i'_g{\models}\bigvee \Sigma \wedge \mathsf{Seg^+}(\rho', i'_{g-1}, i'_g, S'){\in}L(\re'_g)) \wedge \mathsf{Seg^+}(\rho', i'_{n}, i'_{n+1}, S'){\in}L(\re'_{n+1})$ where for any $1\le j \le n+1$, $\re'_j$ is the automata built from $\re_j$ by adding self loop on $\neg \bigvee \Sigma$ (oversampling points) and $S' = S \cup \{\neg \bigvee \Sigma\}$. 
This self loop makes sure that $\re'_j$ ignores(or skips) all the oversampling points while checking for $\re_j$. Hence, $\re'_j$ allows arbitrary interleaving of oversampling points while checking for $\re_j$. Hence, for any $g,h \in dom(\rho)$ with $g',h'$ being old action points of $\rho'$ corresponding to $g,h$, respectively, $\mathsf{Seg^{s}} (\rho, g, h, S){\in}L(\re_i)$ iff  $\mathsf{Seg^{s}} (\rho', g', h', S\cup \{\neg \bigvee \Sigma\}) \in L(\re'_i)$ for $s \in \{+,-\}$. Note that the question, ``$\rho,i\models \phi_f?$'', is now reduced to checking $\mathsf{cond}$  on $\rho'$. \\
 \noindent \textbf{Checking the conditions for $\rho, i \models \phi_f$}. Let $I_g{=}\langle l_g, u_g \rangle$ for any $1 \le g \le n$ (Here, $\langle \rangle$ denotes half-open, closed, or open). We discuss only the case where $\{I_1, \ldots, I_n\}$ are pairwise disjoint and $\inf(I_1) \ne 0$ in $\phi_f=\fregnm(\re_1, \ldots, \re_{n+1})(S)$. 
The case of overlapping intervals can be found in the full version. The disjoint interval assumption  along with [Assumption 1] implies that for any $1 \le g \le n$, $u_{g-1} < l_g$. 
By construction of $\rho'$, between $i'_{g-1}$ and $i'_g$, we have an oversampling point $k_g$. 
 The point $k_g$ is guaranteed to exist between $i'_{g-1}$ and $i'_g$, since these two points lie within two distinct non-overlapping, non-adjacent intervals  $\mathsf{I}_{g-1}$ and $\mathsf{I}_{g}$ from $i'$. Hence their timestamps have different integral parts, and there is always a uniquely labelled oversampling point $k_g$ with timestamp $\lceil \tau'_{i'_{g-1}} \rceil$ between $i'_{g-1}$ and $i'_{g}$ for all $1{\le} g {\le} n$. Let for all $1{\le} g {\le} n+1$, $\re'_g = (Q_g,2^S,init_g, F_g, \delta'_g)$. Let the unique label for $k_g$ be $\nt_{j_g}$.
For any $1\le g \le n$, we assert that the behaviour of propositions in $S$ between points  $i'_{g-1}$ and $i'_g$ (of $\rho'$) should be accepted by  $\re'_g$. This is done by  splitting the run at the oversampling point 
 $k_g$(labelled as $\nt_{j_g}$) with timestamp $\tau'_{k_g}{=}\lceil \tau'_{i'_{g-1}}\rceil$, $i'_{g-1}< k_g < i'_g$. \\
 (1) Concretely, checking for $\mathsf{cond}$,  for each $1 \leq g \leq n$, 
  we start at $i'_{g-1}$ in $\rho'$, from  the initial state $init_g$ of $\re_g$, and move to the state (say $q_g$) that is reached 
 at the closest oversampling point $k_g$. Note that we use only $\re_g$ (we disallow the $\neg \bigvee \Sigma$ self loops) 
 to move to the closest oversampling point. \\ 
(2)  Reaching $q_g$ from $init_g$ we have read a behaviour 
 between $i'_{g-1}$ and $k_g$; this must to the full behaviour, and hence must also be accepted by $\re'_g$(we use $\re'_g$ instead of $\re_g$ to ignore the oversampling points that could be encountered while checking the latter part).
Towards this,  we guess a point  $i'_g$ which is within interval $\mathsf{I}_g$ from $i'$, such that, 
 the automaton $\re'_g$  starts from state $q_g$ reading $\nt_{k_g}$ and reaches a final state in $F_g$ at point  $i'_g$. 
 Then indeed, the behaviour of propositions from $S$ between $i'_{g-1}$ and $i'_g$ respect $\re'_g$, and also 
 $\tau'_{i'_g} -\tau'_{i'} \in \mathsf{I}_g$. \\
 (1) amounts to $\mathsf{Seg^+}(\rho', i'_{g-1}, k_g, S){\in} L(\re_g[init_g,q_g])\cdot \nt_{j_g}$.  This is defined by the formula  $\psi^+_{g-1, \nt_{j_g},Q_g}$
 which asserts $\re_{g+1}[init_{g},q_g] \cdot \nt_{j_g}$ from point $i'_{g-1}$ to the next nearest oversampling point $k_g$ where $\nt_{j_g}$ holds.
 (2) amounts to checking from point $i$, within interval $\mathsf{I}_g$ in its future, the existence of a point $i'_g$ such that $\mathsf{Seg^-}(\rho', i'_g, k_g, S){\in}L(\rev(\nt_{j_g} \cdot \re'_g[q_g,F_g]))$. This is defined by the formula $\varphi^-_{g, \nt_{j_g},q_g}$ which asserts $\rev(\nt_{j_g}\cdot \re'_g[q_g, F_g])$, from point $i'_g$ to an oversampling point $k_g$ which is the earliest oversampling point  s.t. $i'_{g-1}< k_g < i'_g$. 
 For $\mathsf{cond}$, we define the formula 
 $\psi{=}\fut_{[0,1)} \nt_{j_0} \wedge \bigvee \limits_{g{=}1}^{n} [\psi^+_{g-1, \nt_{j_g},q_g} \wedge \psi^-_{g,\nt_{j_g},q_g}] \wedge \psi^+_n$. \\
\noindent $\bullet$ For $1\le g \leq n$, $\psi^+_{g-1, \nt_{j_g},q_g}{=}\fut_{I_g}( \bigvee \Sigma \wedge  \fregm(\re_{g}[init_{g},q_g] \cdot \{\nt_{j_g}\})(S \cup \{\nt_{j_g}\}))$, \\
\noindent $\bullet$ $\psi^+_{n}{=}\fut_{I_n}(\bigvee \Sigma \wedge\fregm(\re_{n+1}\cdot\{\Last\})(S \cup \{\Last\}))$, and \\
\noindent $\bullet$ For $1 \le g \le n$, $\psi^-_{g, \nt_{j_g}, q_g}{=}\fut_{I_g}(\bigvee \Sigma \wedge\sregm(\rev(\nt_{j_g} \cdot \re_{g}[q_g,F_g]))( S \cup \{\nt_{j_g}\}))$.\\ 
Note that there is a unique point between $i'_{g-1}$ and $i'_{g}$ labelled $\nt_{j_g}$. This is because, $\tau'_{i'_g} - \tau'_{i'_{g-1}} < \tau'_{i'_g} - \tau'_{i'} \le \cmax$. 
Hence, we can ensure that the meeting point for the check (1) and (2) is indeed characterized by a unique label.Note that there is exactly one point labeled $\nt_y$ from any point within future $\cmax$ or past $\cmax$ time units (by $\varphi_{\ovs}$). This is the reason we used the counter modulo $\cmax$ to label the oversampling points.  We encourage the readers to see the figure $\ref{fig:pnregmtlmain-sat}$.
\begin{figure}
    \centering\scalebox{0.85}{
    \includegraphics{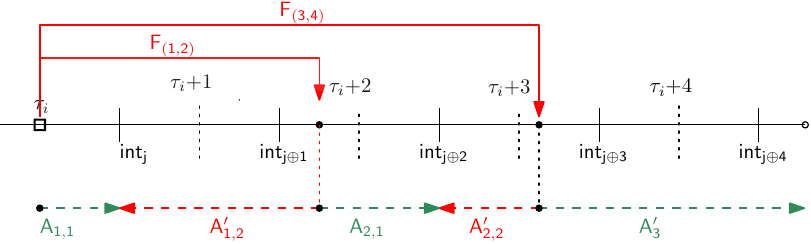}}
    \caption{Figure showing elimination  of $\fregm^2$ modality from temporal definition of the form $\wB(b \leftrightarrow \fregm^2_{(1,2), (3,4)}(A_1, A_2, A_3)(\Sigma')$. This is done by (i)checking for the first part of $A_1$, $A_{1,1}$, from present point to the next oversampling point at timestamp $\lceil \tau_i \rceil$, labelled, $\nt_j$, (ii)jumping to a non-deterministically chosen point within $(1,2)$ and asserting the remaining part of $A_1$ skipping oversampling points, $A'_{1,2}$, in reverse till $\nt_j$, (iii) Following the steps similar to (i) and (ii) for checking $A_2$ but starting the check of first part of $A_2$ from the point chosen in (ii).}
    \label{fig:pnregmtlmain-sat}
\end{figure}
The full $\emitl$ formula $\psi_f$, is obtained by disjuncting over all  $n$ length sequences of states reachable at oversampling points $k_g$ between $i'_{g-1}$ and $i'_g$, and   all  possible values of the unique label $\nt_{j_g}\in \Int$  holding at point $k_g$. 
\\\noindent \textbf{4) Converting the $\emitl$ to $\emitl_{0,\infty}$}: 
We use the reduction from $\emitl$ to equivalent $\emitl_{0,\infty}$ formula \cite{khushraj-thesis}.
In $\psi_f$, only the $\fut$ operators are timed with intervals of the form $\langle l, u \rangle$ where $l>0$ and $u \ne \infty$, but the $\freg$ and $\sreg$ modalities are untimed. We can reduce these time intervals into purely lower bound ($\langle l, \infty) $) or upper bound ($\langle 0, u \rangle$) constraints using these oversampling points, preserving satisfiability, by reduction showed in \cite{khushraj-thesis} Chapter 5 lemma 5.5.2 Page 90-91. 

The above 4 step construction shows that (i) the equisatisfiable $\emitl_{0,\infty}$ formula $\psi$
is of the size $(\mathcal{O}(|\phi|^{Poly(n)})$ where, $n$ is the arity $\phi$. (ii) For a non-adjacent 1-$\tptl$ formula  $\gamma$, applying the reduction in section \ref{sec:tptltopnregmtl} yields $\phi$ of size  $\mathcal{O}(2^{Poly|\gamma|})$ and,  arity of $\phi = \mathcal{O}(|\gamma|^2)$. Also, after applying the reduction of section \ref{sec:pnregsatisfiability} by plugging the value of $|\phi|$ from and its arity from (ii) in (i),  we get the $\emitlinf$ formula $\psi$ of size  $\mathcal{O}(2^{Poly(|\gamma|)*Poly(n)})= \mathcal{O}(2^{Poly(|\gamma|)})$.


\section{Conclusion}
We generalized the notion of non-punctuality  to  non-adjacency in $\tptl$. We proved that satisfiabilty checking for non-adjacent 1-variable fragment of $\tptl$  is EXPSPACE Complete. This gives us a strictly more expressive logic than $\mitl$ while retaining its satisfaction complexity. An interesting open problem is to compare the expressive power of non-adjacent 1-$\tptl$ with that of $\mitl$ with Pnueli modalities (and hence Q2MLO) of \cite{rabinovichY}. We also leave open the satisfiabilty checking problem for non-adjacent $\tptl$ with multiple variables. 
\subsubsection*{Acknowledgements}
{We thank the reviewers of FM 2021 for their feedback to improve the article. This work is partially supported by the European Research Council through the SENTIENT project (ERC-2017-STG \#755953).}

\bibliographystyle{plain}
\bibliography{papers}

\appendix
\newpage
\section{Linear Temporal Logic}
\label{app:lt}
Formulae of $\ltl$ are built over a finite set of propositions $\Sigma$ using Boolean connectives and temporal modalities ($\until$ and $\since$) as follows: $\varphi::=a~|\top~|~\varphi \wedge \varphi~|~\neg \varphi~|
~\varphi \until \varphi ~|~ \varphi \since \varphi$, where  $a \in \Sigma$. 
The satisfaction of an $\ltl$ formula is evaluated over pointed words. For a word $\sigma = \sigma_1 \sigma_2 \ldots \sigma_n \in \Sigma^*$ and a point $i \in dom(\sigma)$, the satisfaction of an $\ltl$ formula $\varphi$ at point $i$ in $\sigma$ is defined, recursively, as follows:
\\(i) \noindent $\sigma, i \models a$  iff  $a \in \sigma_{i}$, (ii) $\sigma,i  \models \neg \varphi$ iff  $\sigma,i \nvDash  \varphi$,
\\(iii) $\rho,i \models \varphi_{1} \wedge \varphi_{2}$  iff  
$\sigma,i \models \varphi_{1}$ 
and $\sigma,i\ \models\ \varphi_{2}$,
\\(iv) $\sigma,i\ \models\ \varphi_{1} \until \varphi_{2}$  iff  $\exists j > i$, 
$\sigma,j\ \models\ \varphi_{2}$, and  $\sigma,k\ \models\ \varphi_{1}$ $\forall$ $i< k <j$,\\
(v) $\sigma,i\ \models\ \varphi_{1} \since_{I} \varphi_{2}$  iff  $\exists j < i$, 
$\sigma,j\ \models\ \varphi_{2}$, and  $\sigma,k\ \models\ \varphi_{1}$ $\forall$ $j< k <i$.

\noindent 
The language of any $\ltl$ formula $\varphi$ is defined as $L(\varphi) = \{\sigma | \sigma, 1\models \varphi\}$.
\section{Examples of non-adjacent specifications}
\label{sec:examp1}
Note that in 1-$\tptl$ we can abbreviate a constraint $T-x \in I$ by $\widehat{I}$.

\begin{example}[non-adjacent 1-$\tptl$] 
An indoor cycling exercise regime  may be specified as follows. One must  {slow pedal} (prop. \emph{sp}) for at least 60 seconds  but until the odometer reads 1km (prop. \emph{od1}). From then onwards one must {fast pedal} (prop \emph{fp}) to a time point in 
the interval [600:900] from the start of the exercise such that pulse rate is sufficiently high
(prop \emph{ph}) for the last 60 seconds of the exercise. This can be given by the following
formula.

\[
x. sp ~~\until~~ \left[
\begin{array}[t]{l}
 ~\widehat{[60,\infty)} ~~\land od1 ~~\land  \\
( ~~fp ~~\until~~ (\widehat{[600:900]} \land x. H(\widehat{[-60,0]} \Rightarrow ph))~~)
\end{array}
\right]
\]

It can be shown that this formula cannot be expressed in logic $\mitl$.
\end{example}

\begin{example}[non-adjacent $\pnregmtl$]
A sugar level test involves the following: A patient visits the lab and is given a sugar measurement test  (prop $sm$) to get fasting sugar level. After this she is given glucose (prop $gl$)
and this must be within 5 min of coming to the lab. After this the patient rests  between 120 and 150 minutes and she is administered sugar measurement again to check the sugar clearance level. After this,
the result (prop $rez$) is given out between 23 to 25 hours (1380:1500 min) of coming to the lab.  We assume that these propositions are mutually exclusive and prop $idle$ denotes negation of all of them.
This protocol is specified by the following non-adjacent $\pnregmtl$ formula. For convenience we
specify the automata by their regular expressions. We follow the convention where the tail automaton $A_{k+1}$ can be omitted in $\fregm^k$.
\[
\fregm^2_{[0,5],~ [1380,1500]} 
   \left[ 
      \begin{array}{l}
         sm \cdot (idle^*) \cdot (gl \land \fregm^1_{[120,150]}(~gl \cdot (idle^*) \cdot sm~), \\
         gl \cdot ((\neg rez)^*) \cdot rez
      \end{array}
    \right]
\]
For readability the two regular expressions of the top $\fregm^2$ are given in two separate lines. It states that the first regular expression must end at time within $[0,5]$ of starting and the second regular expression must end at a time within $[1380,1500]$ of starting. Note the nested use of $\fregm$ in order to anchor the duration between glucose and the second sugar measurement.  
\end{example}

\section{Section \ref{sec:natptl}}
\label{app:exp}

\subsection{Proof of Theorem \ref{thm:me}}
\begin{proof}
Let $\varphi$ be any $\mitl$ formula in negation normal form. Let $TPTL (\varphi)$ be defined as follows.
\begin{enumerate}
    \item $TPTL(a) = a, TPTL(\neg a) = \neg a$ for any $a \in \Sigma$ 
\item $TPTL(\varphi_1 \wedge \varphi_2) = TPTL(\varphi_1) \wedge TPTL(\varphi_2)$, 
\item $TPTL(\varphi_1 \vee \varphi_2) = TPTL(\varphi_1) \vee TPTL(\varphi_2)$ 
\item $\varphi_1 \until_I \varphi_2 = x.((TPTL(\varphi_1)) \until (TPTL(\varphi_2) \wedge x \in I))$, 
\item $\varphi_1 \since_I \varphi_2 = x.((TPTL(\varphi_1)) \since (TPTL(\varphi_2) \wedge x \in I))$ 
\item $\sbf_{[l,u)}(\varphi) = x.\sbf(x\in [0,l) \vee x \in [u,\infty) \vee TPTL(\varphi))$, 
\item $\sbm_{[l,u)}(\varphi) = x.\sbm(x\in [0,l) \vee x \in[u,\infty) \vee TPTL(\varphi))$\footnote{Similar reduction applies for other type of intervals(closed, open left-open right-closed)}
\end{enumerate}

Semantics of $\mtl$ and $\tptl$ imply that for any $\mtl$ formula $\varphi$, $\varphi \equiv TPTL(\varphi)$. 
If $\varphi$ is an $\mitl$ formula, then $TPTL(\varphi)$ is a non-adjacent 1-$\tptl$ formula (when applying 1-5 there will be only one non-punctual interval  and on applying 6 $[0,l)$ and $[u,\infty$ appear as intervals within the scope of clock reset). Hence, every formula in $\mitl$ can be expressed in 1-$\tptl$.
The strict containment is a consequence of Theorem 5 \cite{simoni}. More precisely, \cite{simoni} proves that $\gamma = Fx.(a \wedge \fut(b \wedge x \in (1,2) \wedge T-x \in (1,2) \wedge \fut(c \wedge T-x \in (1,2))))$ is not expressible in $\mtl[\until.\since]$. Note that $\gamma$ is indeed a non-adjacent 1-$\tptl$ formula as there is only one non-punctual interval $(1,2)$ within the scope of the freeze quantifier. 
\end{proof}
\subsection{Figure Illustrating $\pnregmtl$ semantics}
\begin{figure}[h]
\scalebox{0.6}{
    \includegraphics{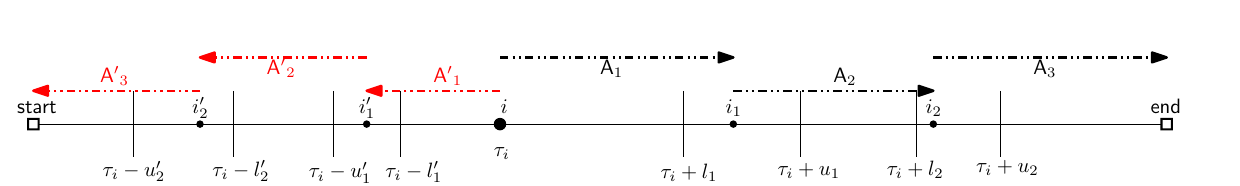}} 
    \caption{Semantics of $\pnregmtl$.$\rho, i{\models}\fregm^{2}_{I_1, I_2}(\mathsf{\re_1, \re_2, \re_3})$ {\&} $\rho, i {\models} \sregm^2_{J_1, J_2}(\mathsf{\re'_1, \re'_2, \re'_3})$ where $I_1 {=} \langle l_1, u_1\rangle, I_2 {=} \langle l_2, u_2\rangle, J_1 {=} \langle l'_1, u'_1\rangle$, $J_2 {=} \langle l'_2, u'_2\rangle$}
    \label{fig:pnregmtl}
\end{figure}

\section{Useful Notations for the rest of the Appendix}
We give some useful notations that will be used repeatedly in the following proofs follows.
\begin{enumerate}
    \item For any set $S$ containing propositions or formulae, let $\bigvee S$ denote $\bigvee \limits_{s\in S}(s)$. Similarly, let $A = \{I_1, \ldots, I_n\}$ be any set of intervals. $\bigcap A = I_1 \cap \ldots \cap I_n, \bigcup A = I_1 \cup \ldots \cup I_n$. For any automaton $A$ let $L(A)$ denote the language of $A$.
 \item For any NFA $A = (Q,\Sigma, i,F,\delta)$, for any $q \in Q$ and 
 $F' \subseteq Q$. $A[q,F'] = (Q,\Sigma, q, F', \delta)$. In other words, $A[q,F']$ is the automaton where the set of states and transition relation are identical to $A$, but the initial state is $q$ and the set of final stats is $F'$. 
 For the sake of brevity, we denote $A[q,\{q'\}]$ as $A[q,q']$. Let $\rev(A) = (Q\cup\{f\},\Sigma, f,\{i\},\delta')$, where $\delta'(f,\epsilon) = F$, for any $a \in \Sigma,q \in Q$, $(q,a,q') \in \delta'$ iff $(q',a,q) \in \delta$. In other words, $\rev(A)$ is an automata that accepts the reverse of the words accepted by $A$.
\item Given any sequence Str, let $|$Str$|$ denote length of the sequence Str. Str[$x$] denotes $x^{th}$ letter of the sequence if $x\le|$Str$|$. Str[$1...x$] denotes prefix of String Str ending at position $x$. Similarly, Str[$x...$] denotes suffix of string starting from $x$ position. Let $S_1, \ldots S_k$ be sets. Then, for any $t \in S_1 \times \ldots \times S_k$ if $t = (x_1, x_2, \ldots, x_k)$. $t(j)$, for any $j < k$, denotes $x_j$. 
\item For a timed word $\rho$, $\rho[i](1)$ gives the set of propositions true at point $i$.  $\rho[i](2)$ gives the timestamp of the point $i$.  
\end{enumerate}

\section{Details  for Section \ref{sec:interval word}}
In this section, we give the missing proofs for all results in section \ref{sec:interval word}. 
\subsection{Proof of Lemma \ref{lem:normalize}}
\label{app: normalize}
We split the proof of Lemma \ref{lem:normalize} into two parts. First, Lemma \ref{c1} shows $\kappa \cong \col(\kappa)$. Lemma \ref{fl2} implies that $\col(\kappa) \cong \normalize(\kappa)$. Hence, both Lemma \ref{c1}, \ref{fl2} together imply Lemma \ref{lem:normalize}.
\begin{lemma}
\label{c1}
Let $\kappa$ be an $I_\nu$-interval word. Then $\kappa \cong \col(\kappa)$.
\end{lemma}
\begin{proof}
A pointed word $\rho,i$ is consistent with $\kappa$ iff
\begin{itemize}
    \item[(i)] $dom(\rho){=}dom(\kappa)$, \item[(ii)]$i{=}\anch(\kappa)$, \item[(iii)] for all $j\in dom(\kappa)$, $\rho[j](1)=\kappa[j]\cap \Sigma$ and \item[(iv)] for all $j \ne i$, $I \in a_j \cap I_\nu$ implies $\rho[j](2) - \rho[i](2)\in I$.
    \item[(v)]$\kappa \sim \col(\kappa)$, by definition of $\col$.
\end{itemize}
 Hence given (v),
(i) iff (a) (ii)iff (b)(iii) iff (c) where:\\
(a) $dom(\rho){=}dom(\kappa){=}dom(\col(\kappa))$, (b) $i{=}\anch(\kappa) = \anch(\col(\kappa))$, (c) for all $j\in dom(\kappa)$, $\rho[j](1)=\kappa[j]\cap \Sigma = \col(\kappa)[j] \cap \Sigma$.
(iv) is equivalent to $\rho[j](2) - \rho[i](2) \in \bigcap(\kappa[j] \cap I_\nu)$, but $\bigcap(\kappa[j] \cap I_\nu) = \col(\kappa)[j]$. Hence, (iv) iff (d) $\rho[j](2) - \rho[i](2) \in \col(\kappa)[j]$.
Hence, (i)(ii)(iii) and (iv) iff (a)(b)(c) and (d). Hence, $\rho,i$ is consistent with $\kappa$ iff it is consistent with $\col(\kappa)$.
\end{proof}
\begin{lemma}
\label{fl2}
Let $\kappa$ and $\kappa'$ be $I_\nu$-interval words such that $\kappa \sim \kappa'$. If for all $I \in I_\nu$, $\first(\kappa, I) = \first(\kappa', I) $ and $\last(\kappa, I) = \last(\kappa', I)$, then $\kappa \cong \kappa'$.
\end{lemma}
\begin{proof}
The proof idea is the following:
\begin{itemize}
    \item As $\kappa \sim \kappa'$, the set of timed words consistent with any of them will have identical untimed behaviour.
    \item As for the timed part, the intermediate $I$-time restricted points ($I$-time restricted points other than the first and the last) do not offer any extra information regarding the timing behaviour. In other words, the restriction from the first and last $I$ restricted points will imply the restrictions offered by intermediate $I$ restricted points.
\end{itemize}
Let $\rho= (a_1, \tau_1),\ldots (a_n,\tau_n)$ be any timed word.
$\rho, i$ is consistent with $\kappa$ iff
\begin{enumerate}
    \item 
    \begin{itemize}
        \item[(i)]$dom(\rho) = dom(\kappa)$, \item[(ii)] $i = \anch(\kappa)$, \item[(iii)] for all $j \in dom(\rho)$, $\kappa[j]\cap \Sigma = a_j$ and \item[(iv)]for all $j\ne i \in dom(\rho)$, $\tau_j -\tau_i  \in \bigcap (I_\nu \cap \kappa[j])$.
    \end{itemize}
    \smallskip 
    
Similarly, $\rho, i$ is consistent with $\kappa'$ if and only if 
\item 
\begin{itemize}
    \item[(a)] $dom(\rho) = dom(\kappa')$, \item[(b)]$i = \anch(\kappa')$, \item[(c)]for all $j \in dom(\rho)$, if $\kappa'[j]  \cap \Sigma = a_j$ and \item[(d)] for all $j\ne i \in dom(\rho)$, $\tau_j -\tau_i  \in \bigcap (I_\nu \cap \kappa'[j])$.

\end{itemize}

\end{enumerate}

Note that as $\kappa \sim \kappa'$, we have, $dom(\kappa) = dom(\kappa')$, $\anch(\kappa) = \anch(\kappa')$, for all $j \in dom (\kappa)$, $\kappa[j]\cap \Sigma = \kappa'[j]\cap \Sigma$. Thus, 2(a) $\equiv$ 1(i), 2(b) $\equiv$ 1(ii) and 2(c) $\equiv$ 1(iii).

Suppose there exists a $\rho, i$ consistent with $\kappa$ but there exists $j' \ne i \in dom(\rho)$, $\tau_j' -\tau_i  \notin I'$ for some $I' \in \kappa'[j']$. By definition, $\first(\kappa',I') \le j' \le \last(\kappa',I')$. But $\first(\kappa',I') = \first(\kappa,I')$, $\last(\kappa',I') = \last(\kappa,I')$. Hence, $\first(\kappa,I') \le j' \le \last(\kappa, I')$. As the time stamps of the timed word increases monotonically, $x \le y \le z$ implies that $\tau_x \le \tau_y \le \tau_z$ which implies that $\tau_x - \tau_i \le \tau_y - \tau_i \le \tau_z -\tau_i$. Hence,  $\tau_{\first(\kappa,I')} - \tau_i \le \tau_{j'} - \tau_i \le \tau_{\last(\kappa,I')}- \tau_i$. But $\tau_{\first(\kappa,I')} - \tau_i \in I'$ and $\tau_{\last(\kappa,I')}- \tau_i \in I'$ because $\rho$ is consistent with $\kappa$. This implies, that $\tau_{j'} - \tau_i \in I'$ (as $I'$ is a convex set) which is a contradiction. Hence, if $\rho,i$ is consistent with $\kappa$ then it is consistent with $\kappa'$ too. By symmetry, if $\rho, i$ is consistent with $\kappa'$, it is also consistent with $\kappa$. Hence $\kappa \cong \kappa'$.
\end{proof}

\section{Details for  Section \ref{sec:tptltopnregmtl}}
In this section, we add all the missing details for section \ref{sec:tptltopnregmtl}. 
\subsection{Proof of Theorem \ref{thm:tptl-ltl}}
\label{app:tptl-ltl}
\begin{proof}
Note that for any timed word $\rho = (a_1,\tau_1) \ldots (a_n,\tau_n) $ and $i \in dom(\rho)$, $\rho,i, [x=:\tau_i] \models \varphi $ is equivalent to $\rho, i \models \psi$.
Let $\kappa$ be any $\I_\nu$-interval word over $\Sigma$ with $\anch(\kappa) = i$. 
\begin{itemize}
    \item [(i)] If $\kappa, i \models \ltl(\varphi)$ then for all 
$\rho \in \mathsf{Time}({\kappa})$ $\rho,i \models \psi$ 
\item[(ii)] If for any timed word $\rho$, $\rho,i \models \psi$ then there exists some $\I_\nu$-interval word over $\Sigma$ such that $\rho, i \in \mathsf{Time}(\kappa)$ and $\kappa, i \models \ltl(\varphi)$. 
Intuitively, this is because $\ltl(\varphi)$ is asserting similar timing constraints via intervals words that is asserted by $\varphi$ on the timed words directly.
\end{itemize}

Formally, (i) and (ii) rely on  Lemma \ref{prop:tptl-ltl}. Substitute $j=i$ and $\gamma = \varphi$ in Lemma \ref{prop:tptl-ltl}.
Hence, $\rho, i \models \psi$ if and only if $\rho, i \in \mathsf{Time}(\{\kappa| \kappa, i \models \ltl(\varphi)\})$. Moreover, $\kappa \in L(\alpha)$ if and only if $\kappa, i \models \ltl(\varphi)$ and $\anch(\kappa) = i$.
\end{proof}

\begin{lemma}
\label{prop:tptl-ltl}
Let $\gamma$ be any subformula of $\varphi$. 
\begin{itemize}
    \item[(i)] For any $\I_\nu$-interval word $\kappa$ and $j \in dom(\kappa)$,  $\kappa,j \models \ltl(\gamma)$ implies for all $\rho, i \in \mathsf{Time}(\kappa)$, $\rho,j, [x =: \tau_i] \models \gamma$.
    \item[(ii)] For every timed word $\rho = (a_1,\tau_1) \ldots (a_n,\tau_n)$ and $j \in dom(\rho)$, $\rho,j, [x =: \tau_i] \models \gamma$ implies there exists an $\I_\nu$-interval word $\kappa$ such that $\rho, i \in \mathsf{Time}(\kappa)$ and $\kappa, j \models \ltl(\gamma)$.
\end{itemize}
\end{lemma}
\begin{proof}
We apply induction on the modal depth of $\gamma$. For depth 0 formulae,  $\gamma$ is a propositional logic formula and  the statement holds trivially. If $\gamma = x \in I$, then $\ltl(\gamma) = I$. For any $\I_\nu$-interval word $\kappa$, if $\kappa, j \models I$ then any $\rho, i$ is consistent with $\kappa$ iff $\tau_j - \tau_i \in I$. This implies, that  $\rho,j, [x = \tau_i] \models \gamma$. Similar argument can be extended to handle the Boolean closure of propositional formulae and clock constraints. 

Assume the above result is true for  formulae of modal depth $k-1$ but not true for some $\gamma$ with modal depth $k > 1$. Hence, $\gamma$ can be written as a Boolean formula over subformulae $\gamma_1, \ldots, \gamma_n$ such that all the formulae are of the modal depth at most $k$, there are no Boolean operators at the top most level of these formulae, and there exists at least one formula $\gamma_a$ of depth $k$ which is a counter-example to the above result. 
Then, $\gamma_a$ is such that:
\begin{itemize}
    \item (i) fails to hold. Let $\kappa,j $ be any arbitrary pointed $\I_\nu$-interval word. $\kappa, j \models \ltl(\gamma_a)$  but there exists $\rho,i\in \mathsf{Time}(\kappa)$ such that $\rho,j, [x =: \tau_i] \models \neg \gamma_a$ or,
    \item (ii) fails to hold. There exists a word $\rho$ such that for $i,j \in dom(\rho)$,  $\rho,j, [x =: \tau_i] \models \neg \gamma_a$ and, there exists an $\I_\nu$-interval word $\kappa$ such that $\rho,i\in \mathsf{Time}(\kappa)$ and
    $\kappa, j \models \ltl(\gamma_a)$.
\end{itemize}
Suppose (i) fails to hold. Then, $\kappa, j \models \ltl(\gamma_a)$  but there exists $\rho,i\in \mathsf{Time}(\kappa)$ such that $\rho,j, [x =: \tau_i] \models \neg \gamma_a$.
Let the outermost modal depth of $\gamma_a$ be $\until$(the reasoning for since modality is similar). Hence, $\gamma_a = \gamma_{a,1} \until \gamma_{a,2}$. $\kappa, j \models \ltl(\gamma_a)$ implies there exists a point $j' >j$ such that $\kappa, j' \models \ltl(\gamma_{a,2})$ and for all points $j<j''<j'$ $\kappa, j'' \models \ltl(\gamma_{a,1})$. 
But as $\gamma_{a_1}$ and $\gamma_{a_2}$ are formulae of depth less than $k$, for every $\rho,i$ consistent with $\kappa$, $\rho, j', [x = \tau_i] \models \gamma_{a,2}$ and for all points $j<j''<j'$ $\rho, j'', [x = \tau_i] \models \gamma_{a,1}$. This implies that for every $\rho, i$ consistent with $\kappa$, $\rho,j, [x = \tau_i] \models \gamma_a$, which is a contradiction. Hence, (i) is true for all subformulae of  $\gamma$.

Suppose (ii) fails to hold. Let $\rho,j, [x=:\tau_i] \models \neg \gamma$ for some $\rho$ and $i,j \in dom(\rho)$. Let $\kappa$ be an $\I_\nu$-interval word such that $\rho,i \in \mathsf{Time}(\kappa)$ and $\kappa \models \ltl(\gamma)$. Let the outermost modal depth of $\gamma_a$ be $\until$ (the reasoning for since is analogous). Hence, $\gamma_a = \gamma_{a,1} \until \gamma_{a,2}$. $\rho,j, [x=:\tau_i] \models \neg \gamma$ implies that 
\begin{itemize}
    \item[(a)]
for all points $j' > j$, $\rho, j', [x=:\tau_i] \models \neg \gamma_{a,2}$, or, 
\item[(b)]  there exists a point $j'> j$ such that $\rho, j' \models \neg(\gamma_{a,2} \wedge \gamma_{a,1})$ and for all points $j<j''<j'$, $\rho, j'' \models \neg \gamma_{a,2}$.
\end{itemize}
 By induction hypothesis, (a) implies there exists a $\I_\nu$-interval word $\kappa'$ such that $\rho, i \in \mathsf{Time}(\kappa')$ and for all points $j' > j$, $\kappa',j' \models \neg \ltl(\gamma_{a,2})$. Similarly (b) implies there exists an $\I_\nu$-interval word $\kappa''$ such that $\rho, i \in \mathsf{Time}(\kappa'')$ and there exists a point $j'> j$ such that $\kappa'', j' \models \neg\ltl(\gamma_{a,2} \wedge \gamma_{a,1})$ and for all points $j<j''<j'$, $\kappa'', j'' \models \neg \ltl(\gamma_{a,2})$. Note that $\kappa' \sim \kappa''$ as $\rho, i$ is consistent with both $\kappa'$ and $\kappa''$. Consider a word $\kappa \sim \kappa' \sim \kappa''$ and $\kappa[i] = \kappa'[i] \cup \kappa''[i]$. By Proposition \ref{prop: unioniw} (below) $\mathsf{Time}(\kappa) = \mathsf{Time}(\kappa') \cap \mathsf{Time}(\kappa'')$. Hence, $\rho, i \in \mathsf{Time}(\kappa)$. Moreover, by Proposition \ref{prop:ltlextended} (below), 
 \begin{itemize}
     \item [(c)] for all points $j' > j$, $\kappa,j' \models \neg \ltl(\gamma_{a,2})$ and
     \item[(d)] there exists a point $j'> j$ such that $\kappa, j' \models \neg\ltl(\gamma_{a,2} \wedge \gamma_{a,1})$ and for all point $j<j''<j'$, $\kappa, j'' \models \neg \ltl(\gamma_{a,2})$.
 \end{itemize}
  Both (c) and (d) implies $\kappa, j'' \models \neg \ltl(\gamma_a)$, which is a contradiction. Hence, (ii) is true for any subformulae $\gamma$ of $\varphi$.
\end{proof}

\begin{proposition}
\label{prop:ltlextended}
Let $\gamma$ be any subformulae of $\varphi$. Let $\kappa, \kappa'$ be any $\I_\nu$-interval words such that $\kappa'\sim \kappa$ and for any $i \in dom(\kappa)$ $\kappa[i]\subseteq \kappa'[i]$. For any $j \in dom(\kappa)$, if $\kappa, j \models \ltl(\gamma)$ then $\kappa', j \models \ltl(\gamma)$.
\end{proposition}
\begin{proof}
Note that $\gamma$ is in negation normal form. Hence, any subformulae of the form $x \in I$ will never be within the scope of a negation. Hence, $\gamma$ can never have a subformulae of the form $\neg (x \in I)$. This implies that $\ltl(\gamma)$ can never have a subformulae of the form $\neg I$ for any $I \in \I_\nu$.
We apply induction on modal depth of $\gamma$. For depth 0 formulae,  $\gamma$ is a propositional logic formula and $\ltl(\gamma)$ is also a propositional logic formula over $\Sigma$ and the statement holds trivially for any pair of similar $\I_\nu$-interval words.
If $\gamma = x \in I$, then $\ltl(\gamma) = I$. If $\kappa, j \models I$ then $I \in \kappa[j]$. This implies that $I \in \kappa' [j]$ (as $\kappa[j] \subseteq \kappa'[j]$). Hence, $\kappa', j\models \ltl(\gamma)$. Similar argument can be extended to handle the Boolean closure of atomic formulae and clock constraints. 

Assume that the above proposition is true for  formulae of modal depth $k-1$ but not true for some $\gamma$ with modal depth $k > 1$. Hence, $\gamma$ can be written as a Boolean formula over subformulae $\gamma_1, \ldots, \gamma_n$ such that all the formulae are of the modal depth at most $k$, there are no Boolean operators at the top most level of these formulae, and there exists at least one formula $\gamma_a$ of depth $k$ such that $\kappa, j \models \ltl(\gamma_a)$ but $\kappa', j$ does not. Let the outermost modality of $\gamma_a$ be $\until$(for other modalities, similar reasoning can be given). Hence, $\gamma_a = \gamma_{a,1} \until \gamma_{a,2}$. This implies that there exists $j' >j$ such that $\kappa, j' \models \ltl(\gamma_{a,2})$ and for all $j < j''<j'$, $\kappa, j'' \models \ltl(\gamma_{a,1})$. But, both $md(\gamma_{a,1})$ and $md(\gamma_{a,2})$ are less than $k$. Hence, by induction hypothesis, $\kappa',j' \models \ltl(\gamma_{a,2})$ and $\kappa', j'' \models \ltl(\gamma_{a,1})$. Hence, $\kappa', j \models \ltl(\gamma_a)$, which is a contradiction. Hence, if $\kappa, j \models \ltl(\gamma)$ then $\kappa', j \models \ltl(\gamma)$.
\end{proof}

\begin{proposition}
\label{prop: unioniw}
Let $\kappa, \kappa', \kappa''$ be $\I_\nu$-interval words such that $\kappa \sim \kappa' \sim \kappa''$ and $\kappa[j] = \kappa'[j] \cup \kappa''[j]$ for any $j \in dom(\kappa)$. Then $\mathsf{Time}(\kappa) = \mathsf{Time}(\kappa') \cap \mathsf{Time}(\kappa'')$.
\end{proposition}
\begin{proof}
We need to prove that $\rho,i \in \mathsf{Time}(\kappa)$ iff $\rho, i \in \mathsf{Time}(\kappa')$ and $\rho, i \in \mathsf{Time}(\kappa'')$.
For any $\rho= (a_1,\tau_1) \ldots (a_n,\tau_n)$ and $i \in dom(\rho)$,\\
$\rho,i \in \mathsf{Time}(\kappa') \cap \mathsf{Time}(\kappa'')$ iff \\
$\forall j \in dom(\rho), a_j = \kappa'[j]\cap \Sigma = \kappa''[j] \cap \Sigma$(as $\kappa' \sim \kappa''$) and $\tau_j - \tau_i \in I$ for all $I \in (\kappa'[j]\cap \I_\nu) \cup (\kappa''[j]\cap \I_\nu)$ iff\\
$\forall j \in dom(\rho), a_j = \kappa[j]\cap \Sigma$(as $\kappa \sim \kappa'' \sim \kappa'$) and $\tau_j - \tau_i \in I$ for all $I \in (\kappa[j]\cap \I_\nu)$(as $\kappa[j]  = \kappa'[j] \cup \kappa''[j]$) iff \\
$\rho,i \in \mathsf{Time}(\kappa)$.
\end{proof}

\subsection{Construction of NFA of type $\seq$}
\label{app:aseqconstruct}
Let $\seq$ be any sequence in $\mathcal{T}(I_\nu)$. Given $A{=}(Q, \init, 2^{\Sigma'}, \delta, F)$ over collapsed interval words from LTL formula $\alpha$. We construct an NFA  $A_\seq{=}(Q \times \{1,2,\ldots |\seq|+1\}\cup \{\bot\}, (\init, 1), 2^{\Sigma'}, \delta_\seq, F \times \{|\seq|+1\})$ such that $L(A_\seq){=}\normalize(L(A) \cap W_\seq)$.

For any $(q,i) \in  Q \times \{1, \ldots, |\seq|+1 \}$,  $S \in 2^{\Sigma \cup I_\nu \cup \anch}$ and $I \in I_\nu \cup \{\anch\}$ such that $\seq[i]{=}I$, $\delta_\seq$ is defined as follows:
\\$\bullet$ If $1\le i \le |\seq|$ 
\begin{itemize}
    \item {\bf(i)} If $\seq[i] \in S$, then $\delta_\seq((q,i), S){=}\delta(q, S) \times \{i+1\}$
\item {\bf(ii)} If $\seq[i] \notin S \wedge S \setminus \Sigma \ne \emptyset$, then $\delta_\seq((q,i), S){=}\emptyset$ 
\item {\bf (iii)} If $S \setminus \Sigma{=}\emptyset$, then $\delta((q,i), S){=}[\bigcup \limits_{I’ \in \I_i} \delta (q,S \cup \{I'\}) \cup \delta_\seq(q, S)] \times \{i\}$ where $\I_i{=}\{I' | I' \in I_{\nu} \wedge \exists i',i''. i' < i \le i''$, $\seq[i'] {=}\seq[i'']{=}I'\}$.
\end{itemize}
$\bullet$ If $i{=}|\seq|{+}1$, $\delta_\seq((q, i), S){=}\emptyset$ if 
$S \setminus \Sigma \ne \emptyset$, 
$\delta_\seq((q,i), S){=}\delta(q, S)\times \{i\}$ if 
$S \setminus \Sigma{=}\emptyset$.
Let $W_\seq$ be all the set of $I_\nu$ intervals words over $\Sigma$ of type $\seq$.

\subsection{Proof of Lemma \ref{lem:nfatonfaseq}}
\label{app:lemnfa}
Let $W_\seq$ be the set of $I_\nu$-interval words of type $\seq$. 

\begin{proof}
\begin{enumerate}
\item (i) Let $w$ be any collapsed interval word of type $\seq$ and $w'{=}\normalize (w)$. Let $\bseq(w){=}\bseq(w')= i_1 i_2 \ldots i_n$ be the boundary positions. Let $i_0 = 0$ and $i_{n+1} = \infty$. Let $j$ be any number such that $i_{k-1} \le j < i_k$. If a state $q$ is reachable by $A$ on reading first $j$ letters of $w$, then $(q, k)$ is reachable by $A_\seq$ on reading the corresponding first $j$ letters of $w'$.

\item (ii) Let $w'$ be any normalized timed word of type $\seq$. Let $\bseq(w')= i_1 i_2 \ldots i_n$ be the boundary positions. Let $i_0 = 0$ and $i_{n+1} = \infty$. Let $j$ be any number such that $i_{k-1} \le j < i_k$.
If a state $(q,k)$ is reachable by $A_\seq$ on reading first $j$ letters of $w'$, then there exists a word $w$ of type $\seq$ such that $w' = \normalize(w)$ and $q$ is reachable by $A$ on reading the corresponding first $j$ letters of $w$.
\end{enumerate}
Statement (i) and (ii) are formally proved in Lemma \ref{lem:AsubsetAseq} and Lemma \ref{lem:AseqsubsetA}, respectively.

(i) implies that on reading any word $w \in W_\seq$, if $A$ reaches the final state then $A_\seq$ reaches the final state on reading $w'{=}\normalize(w)$. Hence, (a) $L(A_\seq) \supseteq \normalize(L(A)\cap W_\seq)$

(ii) implies that on reading any normalized word $w' \in W_\seq$, if $A_\seq$ reaches the final state then there exists a word $w$ accepted by $A$ such that $w = \normalize(w')$(hence, $w \in W_\seq$).  
Hence, (b) $L(A_\seq)\cap \normalize(W_\seq) \subseteq \normalize(L(A)\cap W_\seq)$. By proposition \ref{prop: aseqnorm} (below) we have (c) $L(A_\seq)\cap \normalize(W_\seq) = L(A_\seq)$. Moreover, by (b) and (c) we have (d) $L(A_\seq) \subseteq \normalize(L(A)\cap W_\seq)$. By (a) and (d) we have $L(A_\seq) = \normalize(L(A)\cap W_\seq)$.


\end{proof}

\begin{proposition}
\label{prop: aseqnorm}
$L(A_\seq) \subseteq \normalize(W_\seq)$
\end{proposition}
\begin{proof}
Let $Q_i{=}Q{\times}\{i\}$. By construction of $A_\seq$, transition from a state in $Q_{i}$ to $Q_{i'}$, where $i{\ne}i'$ happens only on reading an interval $I{=}\seq [i]$\footnote{Let $I$ be any symbol in $I_\nu \cup \{\anch\}$. By ``reading of an interval $I$'' we mean ``reading a symbol $S$ containing interval $I$''.}. Moreover, $i'{=}i{+}1$. Thus, any word $w$ is accepted by $A_\seq$ only if
there exists $1{\le} i_1{<}i_2{<}\ldots{<}i_{|\seq|}{\le}|w|$ such that $w[i_k] \setminus \Sigma{=}\{\seq[i_k]\}$ and all other points except $\{i_1,\ldots,i_k\}$ are unrestricted points. This implies, $w{\in}L(A_\seq){\rightarrow}w{\in}\normalize(W_\seq)$.
\end{proof}

Let the set of the states reachable from initial state, $\init$, of any NFA $C$ on reading first $j$ letters of a word $w$ be denoted as $C<w,j>$. Hence, $A<w, 0> = \{\init\}$ and $A_\seq<w,0> = \{(\init, 1)\}$.
\begin{lemma}
\label{lem:AsubsetAseq}
Let $w$ be any collapsed $I_\nu$-interval word of type $\seq$ and $\bseq(w)  = i_1 i_2 \ldots i_n$. Let $i_0 = 0$ and $i_{n+1} = \infty$. Let $w' = \normalize(w)$. Hence, $\bseq(w) = \bseq(w')$. For any $q \in Q$, $q \in A<w, j>$ implies $ (q, k) \in A_\seq<w', j>$ where $i_{k-1} \le j < i_{k}$.
\end{lemma}
\begin{proof}
Recall that $\bseq$ is the sequence of boundary points in order. 
We apply induction on the number of letters read, $j$. Note that for $j = 0$, by definition, $A<w, 0> = \{\init\}$ and $A_\seq(\normalize(w), 0) = \{(\init, 1)\}$ the statement trivially holds as $i_0 \le 0 < i_1 \ldots <i_n$. Let us assume that for some $m$, for every state $q \in A<w, m>$  there exists  $(q, k) \in A_\seq<\normalize(w), m>$ such that $i_1 <\ldots <i_{k-1} \le m < i_k <\ldots i_n$. Now let $j = m+1$. 
Let us assume that $q'$ is any state in $A<w,m+1>$. We just need to show that for some $(q',k') \in A_\seq<w', m+1>$ where $k' = k+1$ if $m+1 \in \boundaryint(w)$. Else $k' = k$. 

As $q' \in A<w, m+1>$, there exists a state $q \in A<w,m>$ such that $q' \in \delta (q, w[m+1])$. By induction hypothesis, $(q, k) \in A_\seq<\normalize(w), m>$.  Note that $(q', k') \in \delta_\seq((q,k), w'[m+1])$ implies $(q',k') \in A_\seq<w', m+1>$.
Let $w[m+1] = S_J$, where $S_J \subseteq \Sigma \cup I_\nu \cup \{\anch\}$ and $S_J \setminus \Sigma$ contains at most 1 element.
\\\textbf{Case 1}: $m+1 \in \boundaryint(w)$. This implies that $w'[m+1] = w[m+1]$. As both $w$ and $w'$ are of type $\seq$, $\{\seq[k]\} = S_J \setminus \Sigma$(by definition of $\seq$). Hence, by construction of $A_\seq$, $\delta_\seq((q,k),S_J) = \delta(q,S_J) \times \{k+1\}$. As $q' \in \delta(q, S_J)$,  $(q',k+1) \in \delta_\seq((q,k),S_J)$.
\\\textbf{Case 2}: $m+1 \notin \boundaryint(w)$. This implies that $w'[m+1] = S = S_J \cap \Sigma $. \\

\noindent \textit{Case2.1}:$S = S_J$. By construction of $A_\seq$, $\delta_\seq ((q,k), S) \supseteq \delta (q, S) \times {k}$. Thus, $(q', k) \in \delta_\seq((q,k),S_J \cap \Sigma)$.\\

\noindent \textit{Case2.2}:$S \ne S_J$. Let $S_J \setminus \Sigma = \{J\}$ where $J \in I_\nu \cup \{\anch\}$. Then $m+1$ is neither the first nor the last $J$-time restricted point nor the anchor point in $w$. Hence, $\first(J, w) < m+1 < \last(J,w)$. By induction hypothesis, $i_{k-1}\le m < i_{k}$. Note, as $m+1$ is not in $\boundaryint(w)$, $m+1 \ne i_k$. Hence, $i_{k-1} \le m < m+1 < i_k$. This implies, $\first (J,w) < i_k \le \last(J,w)$. By definition of $\seq$, there exists $k'$ and $k''$ such that $k' < k \le k''$ and $\seq[k'] = \seq[k''] = J$. Hence, by construction of $\delta_\seq$, $\delta_\seq((q,k), S) \supseteq \delta(a,S_J) \times \{k\}$.
Hence $(q', k) \in \delta_\seq((q,k), S_J)$.
\end{proof}
\begin{lemma}
\label{lem:AseqsubsetA} 
Let $w'$ be any normalized $I_\nu$-interval word of type $\seq$ and \\$\bseq(w')  = i_1 i_2 \ldots i_n$. Let $i_0 = 0$ and $i_{n+1} = \infty$. For any $q \in Q$, $(q,k) \in A_\seq<w', j>$ implies there exists a collapsed $I_\nu$-interval word $w$, such that $\normalize(w) = w'$, $q \in A<w, j>$  and $i_{k-1} \le j < i_k$
\end{lemma}
\begin{proof}
We apply induction on the value of $j$ as in proof of Lemma \ref{lem:AsubsetAseq}. For $j = 0$, the statement trivially holds. Assume that for $j = m$, the statement holds and $(q',k') \in A_\seq<w', m+1>$(Assumption 1). We need to show 
\begin{itemize}
    \item [(A)]$i_{k'-1} \le m+1 <i_{k'}$ and, 
    \item[(B)] there exists $w$ such that $\normalize(w) = w'$ and $q' \in A <w, m+1>$.  
    \end{itemize}
$(q',k') \in A_\seq<w', m+1>$ implies, there exists $(q, k) \in A_\seq<w', m>$ such that $(q',k') \in \delta_\seq((q,k),w'[m+1])$. 

By induction hypothesis, $i_{k-1} \le m < i_{k}$ [IH1] and  there exists a word $w''$ such that $\normalize(w'') = w'$ and $q \in A<w'', m>$ [IH2].

\noindent \textbf{Case 1 } $m+1 \in \boundaryint(w')$: 
This implies  
\begin{itemize}
    \item[(a)] $m+1 \in \{i_1, i_2, \ldots, i_k\}$.
    \item[(b)] $k' = k+1$ (by construction of $\delta_\seq$).
    \item[(c)] $w''[m+1] = w'[m+1] = S \cup \{J\}$ such that $S \subseteq \Sigma $ and $J \in (I_\nu \cup \{\anch\})$.
 In other words, $m+1$ is either a time restricted point or an anchor point in both $w''$ and $w'$. 
 \item[(d)] $\seq[i_k'] = J$, otherwise $\delta_\seq((q,k), S \cup \{J\})) = \emptyset$ contradicting Assumption 1. 
\end{itemize}
Note the following:
\begin{itemize}
    \item [(i)] IH1 and (a) implies that $m+1 = i_k$. This along with b) implies that $m+1 = i_{k'-1}$. Hence proving (A) for Case 1.
\item [(ii)] IH2 along with (c) and (d) implies that $\delta_\seq((q,k), w'[m+1]) = \delta(q, w[m+1]) \times \{k+1\}$. Hence, if $(q',k') \in A_\seq<w', m+1>$ then $q' \in A<w'', m+1>$. Hence, there exists a $w = w''$ such that $q' \in A<w, m+1>$, proving (B) for Case 1.
\end{itemize}

\noindent \textbf{Case 2} $m+1 \notin \boundaryint(w')$ : This implies 
\begin{itemize}
    \item [(1)] $m+1 \notin \{i_1, i_2, \ldots, i_k\}$.
    \item[(2)] $k' = k$ (by construction of $\delta_\seq$). 
    \item[(3)] $w''[m+1] \subseteq \Sigma $. In other words, $m+1$ is an unrestricted point in $w'$ but may or may not be time restricted in $w''$.
\end{itemize}
Now we have 
\begin{itemize}
    \item [(i)]  IH1 implies $i_{k-1} \le m < m+1 \le i_{k}$. This along with (1) and (2) implies $i_{k'-1} \le m < m+1 < i_{k'}$. Hence proving (A) for Case 2.
\item [(ii)] IH2 along with (3) and the construction of $\delta_\seq$ implies $\delta_\seq((q,k), w'[m+1]) = (\bigcup \delta(q, w'[m+1] \cup \{J\}) \cup \delta(q, w'[m+1]))\times k$ for $J \in I_\nu$ such that there exists $j<k'<l$ such that $\seq[j] = \seq[l] = J$. 
\end{itemize}

Hence, $J$ is an interval which appears twice in $\seq$ and only one of those $J$'s have been encountered within first $m$ letters. Hence, the prefix $w'[1...m+1]$ and the suffix $w'[m+2...]$ contains exactly one $J$-time restricted point each. 
This implies that \\

\noindent (Case 2.1) $q' \in \delta_\seq((q,k), w'[m+1] \cup \{J\})$ for some $J$ such that $w'[1...m]$ and $w'[m+2...]$ contains exactly one $J$-time restricted point or \\
\noindent  (Case 2.2) $q' \in \delta_\seq((q,k), w'[m+1])$.\\

As $\normalize(w'') = w'$, first and last $J$-time restricted points are the same in both $w''$ and $w'$. Hence, first $J$-time restricted point in $w''$ is within $w''[1...m]$ and the last is within $w''[m+2...]$. Consider a set of words $W$ such that for any $w \in W$, $w[1...m] = w''[1...m]$, and either $w[m+1] = w'[m+1]$ or $w[m+1] = w'[m+1] \cup \{J\}$ where $J \in I_\nu$ such that both $w'[1...m]$ and $w'[m+2...]$ contains $J$-time restricted points. Notice that $W$ is not an empty set as it will at least contain $w''$ and $w'''$ where $w'''[m+2..]= w''[m+2...]$ and $w'''[m+1] = w'[m+1] \cup \{J\}$ for some $J \in I_\nu$ such that both $w'[1...m]$ and $w'[m+2...]$ contains $J$-time restricted points.
Notice that $m+1 \notin \boundaryint(w)$. Hence, making it time unrestricted will still imply $\boundaryint(w) = \boundaryint(w')$. 

When there exists a $J$ restricted time point in prefix $w[1....m]$ and suffix $w[m+2...]$ for $J \in I_\nu$, making point $m+1$ as $J$ restricted time point will still imply $\boundaryint(w) = \boundaryint(w')$. Hence, this implies that $\normalize(w) = w'$ for any $w \in W$. Moreover, as for any $w \in W$, $w[1...m] = w''[1...m]$, $A<w,m> = A<w'',m>$ and $A_\seq<w,m> = A_\seq<w'',m>$. Hence, for any $q \in Q$ such that $(q,k) \in A_\seq<w',m>$ implies for every $w \in W$, $q \in A<w,m>$.

It suffices to show that there exists a $w\in W$ such that $q' \in A<w,m+1>$. In case of Case 2.1, for any word $w \in W$ such that $w[m+1]$ is a $J$-time restricted point $q' \in A<w,m+1>$. Note that such a word exists as Case 2.1 implies that $w'[1...m]$ and $w'[m+2...]$ contains exactly one $J$-time restricted point. In case of 2.2, for any word $w  \in W$ where $w[m+1] \subseteq \Sigma$, $q' \in A<w,m+1>$. Hence, proving (B) for Case 2.
\end{proof}
\section{Proof of Theorem \ref{thm:natptltonapnregmtl}}
\label{app: tptltopnregmtl}
\begin{proof}
Let $|\psi| = m, |\I_\nu| = n$.
\begin{itemize}
\item Construct an $\ltl$ formula $\alpha$ over interval words such that $\rho,i \models \varphi$ if and only if $ \rho,i \models \mathsf{Time}(L(\alpha))$ as in Section \ref{sec:tptltopnregmtl} bullet \textbf{1a)} such that $|\alpha| = O(n)$.
\item Reduce the $\ltl$ formula $\alpha$ to language equivalent NFA $A'$ using \cite{gastin-oddoux}. This has the complexity $O(2^{n})$. This step is followed by reducing $A'$ to $A$ over interval words over $I_\nu$ such that $L(A) = \col(L(A'))$. Note that the size of $|I_\nu| = O(n^2)$ Section \ref{sec:tptltopnregmtl} bullet \textbf{1b)}.
\item As shown in  bullet \textbf{3)} of Section \ref{sec:tptltopnregmtl} and Lemma \ref{lem:nfatonfaseq}, for any type $\seq$, we can construct $A_\seq$ from $A$ such that $L(A_\seq) = \normalize(L(A_\seq) \cup W_\seq)$ with number of states $k  =  O(2^{Poly(m)})$.
\item As shown in bullet \textbf{4)} of Section \ref{sec:tptltopnregmtl}, for any $\seq$, we can construct $\phi_\seq$ using intervals from $I_\nu$ such that $\rho,i \models \phi_\seq$ iff $\rho, i \in L(A_\seq)$. Note that $Time(L(\varphi)) = Time(L(A)) = \bigcup \limits_{\seq \in \mathcal{T}(I_\nu)} Time(L(A_\seq))$. Note that $|\Sseq| \le (n)^{2n^2} = O(2^{Poly(n)})$. Size of formula $\phi_\seq$ is $(2^{n*m}) \le 2^{m^2}$. Moreover, the arity of the formula $\phi_{\seq} = 2 \times |\seq| = O(2 \times |I_\nu|)$(as each interval from $I_\nu$ appears at most twice in $\seq$)$=O(n^2)$. 
Hence, $\rho, i \models \psi$ if and only if $\rho,i \models \phi$ where $\phi = \bigvee \limits_{\seq \in \Sseq} \phi_\seq$ and the timing intervals used in  $\phi$ comes from $I_\nu$. Note that if $\I$ is non-adjacent than $I_\nu$ is non-adjacent too. Hence, we get a non-adjacent $\pnregmtl$ formula $\phi$ the size of which is $O(2^{Poly(m)})$ and arity is $O(n^2)$. 
\end{itemize}
\end{proof}
\begin{lemma}
\label{thm:simpletptltopnregmtl}
Any simple 1-$\tptl$ formula using intervals from $\I_\nu$ can be reduced to an equivalent $\pnregmtl$ formula using constraints from $I_\nu$ where $I_\nu = \Clcap(\I_\nu)$.  
\end{lemma}
\begin{proof}
The above lemma is a consequence of Theorem \ref{thm:tptl-ltl}, construction from $\ltl$ to NFA \cite{gastin-oddoux}, Lemma \ref{lem:nfatonfaseq} and Lemma \ref{lem:nfatopnregmtl}. 
\end{proof}
\begin{theorem}
\label{thm: tptltopnregmtl}
Any 1-$\tptl$ formula $\psi$ can be reduced to an equivalent $\pnregmtl$ formula.
\end{theorem}
\begin{proof}
Without loss of generality we assume that $\psi$ is of the form $x.\varphi$. We apply induction on the freeze depth of $\varphi$. For $\fd(\varphi) = 0$ the theorem holds due to Lemma \ref{thm:simpletptltopnregmtl}. Assume the theorem holds for $\fd(\varphi) = n$. Consider $\varphi$ with $\fd(\varphi) = n+1$.

Consider every simple subformulae of the form $x.\varphi_i$. Hence, $\fd(\varphi_i) = 0$. Due to Lemma \ref{thm:simpletptltopnregmtl}, we can reduce the formula to an equivalent $\pnregmtl$ formula $\phi_i$. Moreover, if $\psi$ is non-adjacent , then all its subformulae are non-adjacent. Due to Lemma \ref{thm:natptltonapnregmtl}, we have that all $\phi_i$ are non-adjacent and $|\phi_i| = O(2^{Poly(n)})$. Substitute a symbol $a_i$ in place of $x.\varphi_i$ in formula $\psi$. Let the resulting formula be $\psi_a$. Note that $\psi_a$ is of the form $x.\varphi_a$ where $\fd(\varphi_a) = n$. Hence, $\psi_a$ can be reduced to an equivalent $\pnregmtl$ formula $\phi_a$, over $2^{\Sigma\cup A}$ where $A = \{a_1, a_2, \ldots, a_m\}$ and $m$ is the number of simple subformulae of $\psi$. Moreover, if $\psi$ is non-adjacent so is $\phi_a$ and $|\phi_a| = 2^{Poly(|\psi|}$ (by induction hypothesis). As $\pnregmtl$ formulae are closed under nesting, one can substitute back $\phi_i$ in place of $a_i$ to get a formula $\phi$ which is language equivalent to $\psi$.
\end{proof}
\section{Satisfiability Checking for $\pnregmtl$}
\label{app:pnregsatisfiability}
In this section, we show that the satisfiability checking  for non-adjacent $\pnregmtl$ and 1-$\tptl$ is EXPSPACE complete. We show that for any  given non-adjacent $\pnregmtl$ formula $\phi$, we construct an equisatisfiable formula $\emitl_{0,\infty}$ $\psi$ of size $O(2^{Poly(|\phi|)})$. Satisfiability checking for $\emitl_{0,\infty}$ is  PSPACE complete \cite{H19}. Hence, this establishes EXPSPACE upper bound for satisfiability checking problem for non-adjacent $\pnregmtl$. The EXPSPACE lower bound is implied from the lower bound of sublogic $\mitl$.

The rest of the section is dedicated to the construction of equisatisfiable $\emitl_{0,\infty}$ formula $\psi$ given a non-adjacent $\pnregmtl$ formula $\phi$ with at most exponential blow up.


We use the technique of equisatisfiability modulo oversampling \cite{time14}\cite{khushraj-thesis}. Let $\Sigma$ and $\Ovs$ be disjoint set of propositions. 
Given any timed word $\rho$ over $\Sigma$, we say that a word $\rho'$ over $\Sigma \cup \Ovs$ is an oversampling of $\rho$ if $|\rho| \le |\rho'|$ and if you delete the symbols in $\Ovs$ from $\rho'$ you get back $\rho$. Note that when $|\rho'| > |\rho|$, $\rho'$ will have some time points where no proposition from $\Sigma$ is true. These new points are called oversampling points. Moreover, we say that any point $i' \in dom(\rho')$ is an old point of $\rho'$ corresponding to $i$ iff $i'$ is the $i^{th}$ point of $\rho'$ when we remove all the oversampling points. For illustration refer figure \ref{fig:oversample}.
\begin{figure}
    \scalebox{0.63}{
    \includegraphics{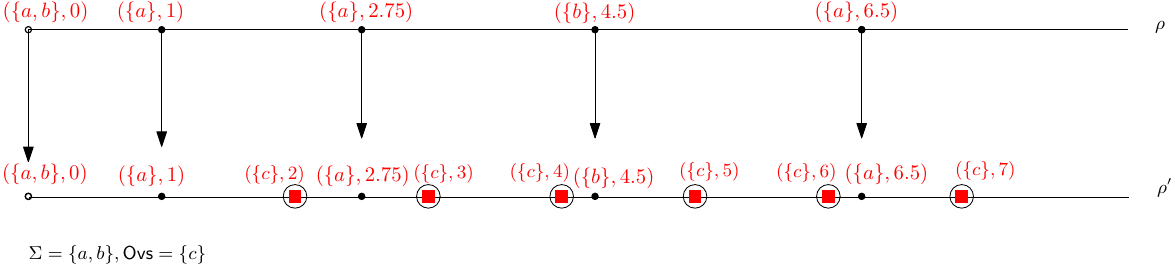}}
    \caption{Figure illustrating oversampling behaviours and projections. $\rho'$ is an oversampling of $\rho$. The points marked with red boxes are oversampling points. The arrow maps an action point of $\rho$ to an old action point of $\rho'$ corresponding to $i$.}
    \label{fig:oversample}
\end{figure}
For the rest of this section, let $\phi$ be a non-adjacent $\pnregmtl$ formula over $\Sigma$. We break down the construction of an $\emitl_{0,\infty}$ formula $\psi$ as follows. 
\\\noindent 1) Add oversampling points at every integer timestamp using formula $\varphi_{\ovs}$,\\
\noindent 2) Flatten the $\pnregmtl$ modalities to get rid of nested automata modalities, obtaining an  equisatisfiable formula $\phi_{flat}$,\\
\noindent 3) With the help of oversampling points, assert the properties expressed by $\pnregmtl$ subformulae $\phi_i$ 
of $\phi_{flat}$ using only $\emitl$ formulae,\\
\noindent 4) Finally, using oversampling points,  get rid of bounded intervals with non-zero lower bound, getting the required $\emitl_{0,\infty}$ formula $\varphi_i$.  Replace  $\phi_i$ in $\phi_{flat}$ with this $\emitl_{0,\infty}$ formula getting $\varphi$.\\
Let $\Last{=}\sbf \bot$ and $\Lastts{=}\sbf \bot \vee (\bot \until_{(0,\infty)} \top)$. $\Last$ is true only at the last point of any timed word. Similarly, $\Lastts$, is true at a point $i$ if there is no next point $i+1$ with the same timestamp $\tau_i$.
Let $\cmax$ be the maximum constant used in the intervals appearing in $\phi$.

\noindent\textbf{1) Behaviour of Oversampling Points}. Let $\Int{=}\{\nt_0,\nt_1,\ldots, \nt_{\cmax-1}\}$  and $\Int \cap \Sigma = \emptyset$.
We oversample  timed words over $\Sigma$ by adding new points where only propositions from $\Int$ holds.  
Given a timed word $\rho$ over $\Sigma$, consider an extension of $\rho$ called $\rho'$, by extending the alphabet $\Sigma$ of $\rho$ to 
$\Sigma' = \Sigma \cup \Int$. Compared to $\rho$, $\rho'$  has extra points called \emph{oversampling} points, where 
 $\neg \bigvee \Sigma$ (and $\bigvee \Int$) hold good. These extra points are chosen at all integer timestamps, in such a way that if $\rho$ already has points with integer time stamps, then the oversampled point with the same time stamp  
 appears last among all points with the same time stamp  in $\rho'$. 
  We will make use of these oversampling points to reduce the $\pnregmtl$ modalities into $\emitl_{0,\infty}$. These oversampling points are labelled with a modulo counter $\Int{=}\{\nt_0,\nt_1,\ldots, \nt_{\cmax-1}\}$. 
The counter is initialized to be $0$ at the first oversampled point with timestamp $0$ and is incremented, modulo $\cmax$, after exactly one time unit till the last point of $\rho$.  
Let $i \oplus j{=}(i+ j) \% \cmax$, where $\%$ is a modulo operator. 
The  oversampled behaviours are expressed using the formula $\varphi_{\ovs}$:  
$\{\neg \fut_{(0,1)} \bigvee \Int \wedge \fut_{[0,1)} \nt_{0}\} \wedge $ \\
$ \{\bigwedge \limits_{i=0}^{\cmax-1}\wB\{(\nt_i \wedge \fut (\bigvee \Sigma))\rightarrow (\neg \fut_{(0,1)} (\bigvee \Int) \wedge \fut_{(0,1]} (\nt_{i \oplus 1} \wedge(\neg \bigvee \Sigma) \wedge \Lastts))\}$. 
Let $\ext: T\Sigma^*\longrightarrow T\Sigma'^* $ map timed words $\rho$  
to a word $\rho'$ over $\Sigma'$ given by $\ext(\rho){=}\rho'$ iff \textbf{(i)}$\rho'$ is an oversampling of $\rho$ and \textbf{(ii)}$\rho' \models \varphi_{\ovs}$. Map  $\ext$ is well defined : for any word $\rho$, there is a unique $\rho'$ over $\Sigma'$ satisfying (i) and (ii). Moreover, for any word $\rho'$ satisfying $\varphi_{\ovs}$ there exists a unique word $\rho$ such that $\ext(\rho) = \rho'$.

\noindent \textbf{2) Flattening}. Next, we  flatten $\phi$ to eliminate the nested $\fregkm$ and $\sregkm$ modalities while preserving satisfiability. Flattening is well studied  \cite{deepak08}, \cite{time14}, \cite{khushraj-thesis}, \cite{H19}.
The idea is to associate a fresh witness variable $b_i$ to each subformula  $\phi_i$ which needs to be eliminated. 
This is achieved using the \emph{temporal definition} $T_i=\wB((\bigvee \Sigma \wedge \phi_i) \leftrightarrow b_i)$  and replacing 
$\phi_i$ with $b_i$ in $\phi$,  $\phi''_i{=}\phi[b_i /\phi_i]$, where  $\wB$ is the weaker form of 
$\sbf$ asserting at the current point and strict future.   Then, $\phi'_i{=}\phi''_i \wedge T_i \wedge \bigvee\Sigma$ is equisatisfiable to $\phi$. 
Repeating this across all  subformulae of $\phi$, we obtain $\phi_{flat}{=}\phi_t \wedge T \wedge \wB(\bigvee W \rightarrow \bigvee \Sigma)$ 
over the alphabet  $\Sigma'{=}\Sigma \cup W$, where $W$ is the set of all fresh variables, $T=\bigwedge_i T_i$, $\phi_t$ 
is a propositional logic formula over $W$.  Each $T_i$ is of the form $\wB(b_i \leftrightarrow (\phi_f \wedge \bigvee\Sigma))$ where $\phi_f{=}\fregnm(\re_1, \ldots, \re_{n+1})(S)$ (or uses $\sregnm$) and $S \subseteq \Sigma'$. Note that the flattening results in $S$ being a set  of propositional logic formulae over $\Sigma'$. We eliminate the Boolean operators by further flattening those and replacing it with a single witness proposition; thus $S \subseteq \Sigma'$.  This  flattens $\phi$, and  all the $\fregnm$ and $\sregnm$ subformulae appearing in $\phi_{flat}$ are of modal depth 1.

For example, consider the formula $\phi=\fregm^2_{(0,1)(2,3)}(A_1, A_2,A_3)(\{\phi_1 , \phi_2\})$, where $\phi_1 = \sregm^2_{(0,2)(3,4)}(A_4,A_5, A_6)(\Sigma),\phi_2 = \sregm^2_{(1,2)(4,5)}(A_7,A_8, A_9)(\Sigma)$.   
Replacing the $\phi_1, \phi_2$ modality with witness propositions 
$b_1, b_2$, respectively, we get \\$\phi_t=\fregm^2_{(0,1)(2,3)}(A_1',A_2',A_3')
) (\{b_1, b_2\}) \wedge T$, 
where
$T=\wB(b_1 \leftrightarrow (\bigvee \Sigma \wedge  \phi_1)) \wedge \wB(b_2 \leftrightarrow (\bigvee \Sigma \wedge \phi_2))$, $A_1',A'_2,A'_3$ are automata constructed from $A_1, A_2, A_3$, respectively by replacing $\phi_1$ by $b_1$ and $\phi_2$ by $b_2$ in the labels of their transitions . Hence, $\phi_{flat}=\phi_t \wedge T$ is obtained by flattening the $\fregkm,\sregkm$ modalities from $\varphi$.

Without losing generality, we make following assumptions. 

\noindent \textbf{Assumption 1:} All the subformulae of the given $\pnregmtl$ formula, $\phi$, of the form $\fregkm(\re_1,\ldots, \re_{k+1})(S)$ and $\sregkm(\re_1,\ldots, \re_{k+1})(S)$ is such that $\inf(\mathsf{I_1}) \le \inf(\mathsf{I_2}) \le \ldots \le \inf(\mathsf{I_n})$ and $\sup(\mathsf{I_1}) \le \sup(\mathsf{I_2}) \le \ldots \le \sup(\mathsf{I_n})$. This is because check for $\re_{i+1}$ cannot start before check of $\re_{i}$ in case of $\fregkm$ modality (and vice-versa for $\sregkm$ modality) for any $1\le i \le k$ . More precisely, for any $1\le i <  j k$,  let $\mathsf{I}_i = [l_i, u_i), \mathsf{I}_j = [l_j, u_j)$ and $l_i > l_j$, then   $\fregkm(\re_1,\ldots, \re_{k+1})(S) \equiv \freg^{\mathsf{k}}_{\mathsf{I_1,\ldots,I'_j\ldots I_k}}(\re_1,\ldots, \re_{k+1})(S)$, where $\mathsf{I}'_j = [l_i, u_j)$ (the argument can be similarly generalized for other type of interval for the case where $u_i > u_j$ and for $\sregkm$ modality). 

\noindent \textbf{Assumption 2:} Intervals $I_1, \ldots I_{k-1}$ are bounded intervals. Interval $I_k$ may or may not be bounded \footnote{Applying the identity $\fregm^k_{I_1, I_2, \ldots [l_1, \infty) [l_2, \infty)}(\re_1, \ldots, \re_{k+1}){\equiv}\\\fregm^k_{I_1, I_2, \ldots [l_1, \cmax) [l_2, \infty)}(\re_1, \ldots, \re_{k+1}) {\vee} \fregm^k_{I_1, I_2, \ldots  [l_2, \infty)}(\re_1, \ldots, \re_{k-1},\re_{k} \cdot \re_{k+1})$ we can get rid of intermediate unbounded intervals.}.

\noindent \textbf{3)Obtaining  equisatisfiable $\emitl$ formula $\psi_f$  for the $\pnregmtl$ formula $\phi_f$ in 
each $T_i=\wB(b_i \leftrightarrow (\phi_f \wedge \bigvee\Sigma))$}. The next step is to replace all the $\pnregmtl$ formulae occurring in temporal definitions $T_i$. 
We use oversampling to construct the formula $\psi_f$ : for any timed word $\rho$ 
over $\Sigma$, $i \in dom(\rho)$, there is an extension $\rho'=\ext(\rho)$ over 
an extended alphabet $\Sigma'$, and a point $i' \in dom(\rho')$ which is an old point 
corresponding to $i$ such that 
$\rho', i' \models \psi_f$ iff $\rho, i \models \phi_f$.  
\smallskip 
Consider $\phi_f=\fregnm(\re_1, \ldots, \re_{n+1})(S)$ where $S \subseteq \Sigma'$. 

\noindent Let $\rho =(a_1, \tau_1) \dots (a_n, \tau_n)$ be a timed word over $\Sigma$, $i \in dom(\rho)$. Let $\rho' = \ext(\rho)$ be defined by 
$(b_1, \tau'_1)\dots (b_m, \tau'_m)$ with $m \geq n$, and each $\tau'_i$ is a either a new integer timestamp 
not among $\{\tau_1, \dots, \tau_n\}$ or is some $\tau_j$ if $i$ is an old point corresponding to $j$. 
Let  $i'$ be an old point in $\rho'$ corresponding to $i$. Let $i'_0{=}i'$ and $i'_{n+1} = |\rho'|$. 
 $\rho, i \models \phi_f$ iff Condition C = $\exists i'\le i_1' \le \ldots \le i'_{n+1} \bigwedge \limits_{g{=}1}^n (\tau'_{i'_g} -\tau'_{i'} {\in} \mathsf{I}_g \wedge \rho',i'_g{\models}\bigvee \Sigma \wedge \mathsf{Seg^+}(\rho', i'_{g-1}, i'_g, S'){\in}L(\re'_g)) \wedge \mathsf{Seg^+}(\rho', i'_{n}, i'_{n+1}, S'){\in}L(\re'_{n+1})$ where for any $1\le j \le n+1$, $\re'_j$ is the automata built from $\re_j$ by adding self loop on $\neg \bigvee \Sigma$ (oversampling points) and $S' = S \cup \{\neg \bigvee \Sigma\}$ 
This self loop makes sure that $\re'_j$ ignores(or skips) all the oversampling points while checking for $\re_j$. Hence, $\re'_j$ allows arbitrary interleaving of oversampling points while checking for $\re_j$. Hence, for any $g,h \in dom(\rho)$ with $g',h'$ being old action points of $\rho'$ corresponding to $g,h$, respectively, $\mathsf{Seg^{s}} (\rho, g, h, S){\in}L(\re_i)$ iff  $\mathsf{Seg^{s}} (\rho', g', h', S\cup \{\neg \bigvee \Sigma\}) \in L(\re'_i)$ for $s \in \{+,-\}$. Note that we reduced the problem `$\rho,i \models \phi_f$?' into a problem to check a similar property on $\rho',i'$ (Condition C).
\smallskip 

\noindent \textbf{Checking the conditions for $\rho, i \models \phi_f$}. Let $I_g{=}[ l_g, u_g )$\footnote{Similar reasoning will hold for other type of intervals and their combination} for any $1 \le g \le n$. 
We first discuss the simple case, \textbf{Non-Overlapping Case}
where $\{I_1, \ldots, I_n\}$ are pairwise disjoint and $\inf(I_1) \ne 0$ in $\phi_f=\fregnm(\re_1, \ldots, \re_{n+1})(S)$. The more general case of overlapping intervals is discussed later.
The disjoint interval assumption along with [Assumption 1] implies that for any $1 \le g \le n$, $u_{g-1} < l_g$. 
Between $i'_{g-1}$ and $i'_g$, we have an oversampling point $k_g$. 
 The point $k_g$ is guaranteed to exist between $i'_{g-1}$ and $i'_g$, since these two points lie within two distinct non-overlapping, non-adjacent intervals  $\mathsf{I}_{g-1}$ and $\mathsf{I}_{g}$ from $i'$. Hence their timestamps have different integral parts, and there is always a uniquely labelled oversampling point $k_g$ between $i'_{g-1}$ and $i'_{g}$, labelled $j_g$ for all $1\le g \le n$. Let for all $1\le g \le n+1$, $\re'_g = (Q_g,2^S,init_g, F_g, \delta'_g)$. 
 
Formally, $\tau'_{i'_{g-1}} - \tau'_{i'} \le u_{g-1} < l_g  \le  \tau'{i'_{g}} - \tau'_{i'}$(By (i)). As the $l_g$ and $u_{g-1}$ are integers, $u_{g-1} +1 \le l_g$. Hence, $\tau'{i'_{g}} - \tau'{i'_{g-1}} \ge 1$. 

For any $1\le g \le n$, we assert that the behaviour of propositions in $S$ between points $i'_{g-1}$ and $i'_g$ should be accepted by the automaton $\re'_g$. This is done by  splitting the run at the oversampling point $k_g$ with timestamp $\tau'_{k_g}{=}\lceil \tau'_{i'_{g-1}}\rceil$, lying in between  $i'_{g-1}$ and $i'_g$. \\
 (1) The idea is to start from the initial state $init_g$ of $\re_g$, and move to the state (say $q$) that is reached 
 at the closest oversampling point $k_g$. Note that we use only $\re_g$ (we do not allow self loops on $\bigvee \Sigma$) 
 to end up at the closest oversampling point.  Let $\nt_{j_g}$ be the symbol from $\Int$ decorating $k_g$ : this is checked 
 by $\rho', i'_{g-1}  \models \fut_{[0,1)} \nt_{j_g}$. \\
(2)  If we can guess  point $i'_g$ which is within interval $\mathsf{I}_g$ from $i'$, such that, 
 the automaton $\re'_g$  starts from state $q$ reading $\nt_{j_g}$ and reaches a final state in $F_g$ at point 
 $i'_g$, then indeed, the behaviour of propositions from $S$ between $i'_{g-1}$ and $i'_g$ respect $\re'_g$, and also 
 $\tau'_{i'_g} -\tau'_{i'} \in \mathsf{I}_g$. \\
 (1) amounts to $\mathsf{Seg^+}(\rho', i'_{g-1}, k_g, S){\in} L(\re_g[init_g,q])\cdot \nt_{j_g}$.  This is 
 defined by the formula  $\psi^+_{g-1, \nt_{j_g},Q_g}$
 which asserts $\re_{g+1}[init_{g},q] \cdot \nt_{j_g}$ from point $i'_{g-1}$ to the next nearest oversampling point $k_g$ where $\nt_{j_g}$ holds.
 (2) amounts to checking from point $i$, within interval $\mathsf{I_g}$ in its future, the existence of a point $i'_g$ such that $\mathsf{Seg^-}(\rho', i'_g, k_g, S){\in}L(\rev(\nt_{j_g} \cdot \re'_g[q,F_g]))$. This is defined by the formula $\psi^-_{g, \nt_{j_g},Q_g}$ which asserts $\rev(\nt_{j_g}\cdot \re'_g[q, F_g])$, from point $i'_g$ to an oversampling point $k_g$ which is the earliest oversampling point  s.t. $i'_{g-1}< k_g < i'_g$. 
 For $\rho', i' \models \phi_f$, we define the formula 
 $\psi{=}\fut_{[0,1)} \nt_{j_0} \wedge \bigvee \limits_{g{=}1}^{n} [\psi^+_{g-1, \nt_{j_g},Q_g} \wedge \psi^-_{g,\nt_{j_g},Q_g}] \wedge \psi^+_n$. 
 
 \textbf{Note that there is a unique point between $i'_{g-1}$ and $i'_{g}$ labelled $j_g$. This is because, $\tau'_{i'_g} - \tau'_{i'_{g-1}} < \tau'_{i'_g} - \tau'_{i'} \le \cmax$. As the oversampling points are labelled by counter which resets only after $\cmax$ time units, there can not be two different oversampling points labelled with the same counter value}. Hence, we can ensure that the meeting point for the check (1) and (2) is indeed marked with a unique label.

We define a set of sequences $\Dseq{=}\{x_1 x_2 \ldots x_n| 1\le g \le n$, $x_g \in Q_g\}$. 
Note that $|\Dseq|{=}|Q_1| \times \ldots \times |Q_{n+1}|{=}O(2^{Poly(|\phi_f|)})$ . Each sequence $\dseq = q_1 q_2 \ldots q_n {\in} \Dseq$ defines a subcase where the part of accepting run between $i'_{g-1}$ and next nearest oversampling point $k_g$ ends at state $q_g$. 
Similarly, we define a set $\Tseq$ containing sequences of the from $y_1 \ldots y_n$ such that $0 \le y_1 \le \cmax-1$ and for any $2\le g \le n$, $y_g{=}\{ y \oplus y_1 | l_g \le y \le u_g\}$. Hence, $|\Tseq|{=}\cmax \times |I_1| \times |I_2| \times |I_{n-1}|$. 
Intuitively, a sequence $\tseq =  j_1 j_2 \ldots j_{n} {\in} \Tseq$ identifies a subcase where the point next closest oversampling point from $i'_{g-1}$ ($k_g$) is labelled as $\nt_{j_g}$ for any $1 \le g \le n-1$. In other words, $\rho', i'_{g-1} \models \fut_{[0,1)} \nt_{j_g}$. Note that, due to [Assumption 2], $\mathsf{I_1, \ldots, I_{n-1}}$ are necessarily bounded intervals. This implies that all the points $i_1' \ldots, i'_{n-1}$ appears within $\cmax$ time units from $i'$. As the counter $\Int$ is modulo $\cmax$, for any $1\le g<g' \le n-1$, $j_g' \ne j_g$. 
Given $\tseq {\in} \Tseq$ and $\dseq {\in} \Dseq$, for any $g\le n$, condition $(\tau'_{i'_g} - \tau'_{i'} {\in} \mathsf{I}_g \wedge \rho,i'_g{\models}\bigvee \Sigma \wedge \mathsf{Seg^+}(\rho', i'_{g-1}, i'_g, S'){\in}L(\re'_g))$ is equivalent to (1) $\mathsf{Seg^+}(\rho', i'_{g-1}, k_g, S'){\in}(L(\re_g[init_g,q_g]\cdot \nt_{j_{g}})$ and, (2) from point $i$, within interval $\mathsf{I_g}$ in future, there exists a point $i'_g$ such that $\mathsf{Seg^-}(\rho', i'_g, k_g, S'){\in}L(\rev(\nt_{j_g} \cdot \re'_g[q_g,F_g]))$. 
We define a formula \\$\psi_{\tseq, \dseq}{=}\fut_{[0,1)} \nt_{j_1} \wedge \bigvee \limits_{g{=}1}^{n} [\psi^+_{g-1, \nt_{j_g},q_g} \wedge \psi^-_{g, \nt_{j_g},q_g}] \wedge \psi^+_n$.
\\$\psi^+_{g-1, \nt_{j_g},q_g}$ asserts $\re_{g}[init_{g},q_g] \cdots \nt_{j_g}$ from point $i'_{g-1}$ to the next nearest oversampling point $k_g$ where $\nt_{j_g}$ holds, hence equivalent to (a) . Similarly, $\psi^-_{g, j_g,q_g}$ formula asserts $\rev(\nt_{j_g}\cdot \re_g[q_g, F_g])$, from point $i'_g$ to an oversampling point $k_g$ which is earliest oversampling point  s.t. $i'_{g-1}< k_g < i'_g$.

$\psi^+_{0, \nt_{j_1},q_1}{=} (\bigvee \Sigma \wedge  \fregm(\re_{1}[init_{1},q_1] \cdot \{\nt_{j_1}\})(S \cup \{\nt_{j_1}\}))$

For $2 \le g \le n$,
$\psi^+_{g-1, \nt_{j_g},q_g}{=}\fut_{I_{g-1}}( \bigvee \Sigma \wedge  \fregm(\re_{g}[init_{g},q_g] \cdot \{\nt_{j_g}\})(S \cup \{\nt_{j_g}\}))$

$\varphi^+_{n}{=}\fut_{I_n}(\bigvee \Sigma \wedge\fregm(\re_{n+1}\cdot[\Last])( S \cup \{\Last\}))$.

For $1 \le g \le n$
\\$\varphi^-_{g, j_g, q_g}{=}\fut_{I_g}(\bigvee \Sigma \wedge\sregm(\rev(\nt_{j_g} \cdot \re_{g}[q_g,F])) (S \cup \{\nt_{j_g}\} ))$.

Note that there is exactly one point labeled $\nt_{j_g}$ from any point within future $\cmax$ or past $\cmax$ time units(by $\varphi_{\ovs}$). We encourage the readers to see the figure $\ref{fig:pnregmtl-sat}$.
\begin{figure}
    \centering\scalebox{0.85}{
    \includegraphics{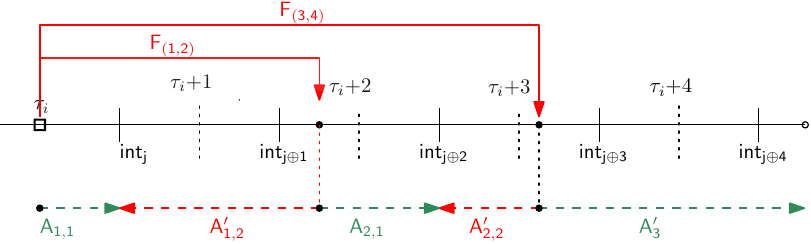}}
    \caption{$\re_{i,1} = \re_1[init_i, x_i] \circ \nt_{y_i}, \re'_{1,2} = \nt_{y_i} \circ \re'_2[init_2, x_2]$ for $i{\in}\{1,2\}$, $\dseq = x_1 x_2$, $y_1 = j, y_2 = j\oplus 2$. Hence, $\tseq = j\ j{\oplus}2$}
    \label{fig:pnregmtl-sat}
\end{figure}
Finally disjunction over all $\tseq {\in} \Tseq$ and $\dseq{\in} \Dseq$ we get the formula $\varphi'_f$.
\\$\varphi'_f{=}\bigvee \limits_{\tseq {\in} \Tseq, \dseq {\in} \Dseq} \varphi_{\tseq, \dseq}$
\subsection{General Case}
 We now discuss the more general case where the intervals can have overlaps.  As we assumed there were no overlaps amongst intervals in the previous case, we were able to guarantee that $\tseq[g]<\tseq[g+1]$. In other words, if for some $1 \le g \le n$ if $I_{g-1} \cap I_g \ne \emptyset $ we can not guarantee existence of an oversampling point $k_g$ between $i'_{g-1}$ and $i'_{g}$. This also implies that their could be a sequences $\tseq \in \Tseq$ $\tseq = j_1\ldots j_n$ such that $j_{g}=j_{g+1}$ for some $1\le g \le n$. 
All these points are within the scope of a time constrained existential quantifier with a restriction on sequence of points (i.e., $i'\le i_1 \le \ldots \le i_n$). In the previous case, we express this specification by asserting the required $\fregm$ formulae within the scope of $\fut_{I_k}$ operator for $1\le k \le n$. This was enough in the previous case because every point within $I_k$ time from $i'$ appears strictly before every point within time $I_{k+1}$ from $i'$ as these intervals are non-adjacent and all the intervals were assumed to be pair-wise disjoint. Hence, the natural ordering of intervals implicitly asserted the required restriction on the ordering of points. While we still have non-adjacency, we no more have such an assumption on disjointedness of intervals $\mathsf{I_0, I_1, \ldots, I_n}$. Hence, $\fut_{\mathsf{I}}$ interval doesn't seem to be enough to assert these restrictions. For example, consider $\mathsf{I_1} = (1,2)=\mathsf{I_2}$. It is not possible to make sure that $\fut_{(1,2)}$ modality in $\psi^+(1,-,-)$ chooses a point that appears before the point chosen by $\fut_{(1,2)}$ modality in $\psi^-(2,-,-)$. This precise property makes $\mtl$ strictly less expressive than logics like $\tptl$.
But we can still assert this time constrained existential quantification with the restriction on the sequence of points using the following observation based on non-adjacency. 
\begin{proposition}
\label{prop:non-adjacent}
Let $\I{=}\{\langle l_1, u_1\rangle , \ldots, \langle l_n, u_n\rangle\}$ be non-adjacent set of intervals in $\intintervaln$ and $t'$ be any real number. Then, there does not exist $ 1 \le i \le j \le n$ such that $t< t'+ u_i\le t+1$ and $t< t'+ l_{j} \le t+1$.
\end{proposition}
\begin{proof}
Suppose there exists $u_i$ and $l_j$ such that $t< t'+ u_i \le t+1$ and $t< t'+ l_{j}\le  t+1$. This implies, that $|u_i - l_j| < 1$. But $u_i$ and $l_j$ are integers. Hence, $u_i{=}l_j$. This implies that $\I$ is adjacent, which is a contradiction.
\end{proof}
Let $t_0$ be timestamp of $i'$. Let for some $0\le g \le n$, $1\le h \le n-g-1$, $j_{g}<j_{g+1}{=}j_{g+2}{=}\ldots =j_{g+h+1}<j_{g+h+2}$\footnote{let $j_0 = -1$ and $j_{n+1} = \infty$}. Such a string corresponds to a case where points $i'_{g}, \ldots i'_{g+h}$ have timestamps in $(t,t+1]$ from the first point for some integer $t$. Hence, there is no oversampling point appearing between these points. This is only possible if $t_0 + \mathsf{I}_{g+k} \cap (t,t+1] \ne \emptyset$ for $1\le k \le h+1$. 
Also note that if $(t,t+1] \subseteq t_0 + \mathsf{I}_{g+k}$ for some $1\le k \le h+1$ then we do not need to assert any time constrained for $i'_{g+k}$ as its occurrence within time interval $(t, t+1]$ implies the required time constraint for $i'_{g+k}$. By [Assumption 1] $u_g$ (and $l_{g+h}$) is the smallest amongst upper-bounds (largest amongst lower bounds, respectively) appearing in  $\mathsf{I_{g}, \ldots I_{g+h}}$. By proposition \ref{prop:non-adjacent}, only one of (A)$t < t_0 + u_g \le t+1$ or (B)$t < t_0 + l_{g+h} \le t+1$ is true. 
\begin{itemize}
    \item Case 0: Neither (A) nor (B) is true. In this case, $(t,t+1] \subseteq t_0 + I_{g+k}$ for any $0\le k \le h$. Hence, there is no need to assert any timing requirements.
     \item Case 1: Only (A) holds. Let $0 \le k \le h$ be maximum number such that $u_{g+k}{=}u_g$. Then, $i'_{g+k} {\in} I_k$ implies for all $0\le k' \le k$ $i'_{g+k'} {\in} I_{k'}$. Moreover, for all $k < k'' \le h$, $(t,t+1] \subseteq t_0 + I_{g+k''}$. Hence, $i'_{g+k} {\in} I_k$ implies for all $0\le k' \le h$ $i'_{g+k'} {\in} I_{k'}$
     \item Case 2: Only (B) holds. Let $0 \le k \le h$ be minimum number such that $l_{g+k}{=}l_{g+h}$. Then, $i'_{g+k} {\in} I_k$ implies for all $k\le k' \le h$ $i'_{g+k'} {\in} I_{k'}$. Moreover, for all $0 \le k'' < k$, $(t,t+1] \subseteq t_0 + I_{g+k''}$. Hence, $i'_{g+k} {\in} I_k$ implies for all $0\le k' \le h$ $i'_{g+k'} {\in} I_{k'}$
 \end{itemize}
Note that Case 0 holds if and only if, $u_g > (j_{g+1} - j_1)\%\cmax$ and $l_{g+h} \le  (j_{g+1} - j_1)\%\cmax$. Case 1 holds when $(j_{g+1} - j_1)\%\cmax- 1<u_g \le (j_{g+1} - j_1)\%\cmax$. Similarly, Case 2 holds when $(j_{g+1} - j_1)\%\cmax+1 \ge l_{g+h} > (j_{g+1} - j_1)\%\cmax$. Note that given a $\tseq$ and $\mathsf{I_1, \ldots , I_k}$, $k$ can be determined uniquely. Let $\dseq = q_1 q_2 \ldots q_n$. 
Hence the following formulae are asserted depending on the conditions that hold:
 $\psi_{0, \tseq, \dseq,g,h}{=}\neg \nt_{j_{g+1}-1} \until (\nt_{j_{g+1}-1} \wedge \fregm(\re_g[q_g,F_g]\cdot \re_{g+1}.\cdots.\re_{g+h}\cdot\re_{g+h+1}[init_{g+h+1}, q_{g+h+1}]\cdot \{\nt_{j_{g+1}}\})( S \cup \{\nt_{j_{g+1}}\}))$
 In the case 0, we start from point $i'$, reach the first oversampling point $k$ labelled $\nt_{j_{g+1}-1}$. Note that all the points $i'_g$ to $i'_{g+h}$ are within interval $(\tau'_k, \tau'_k+1]$. Moreover, the required timing constraints imposed by $\phi_f$ is implied by the existence of these points within interval $(\tau'_k, \tau'_k+1]$. Hence, we just need to assert the specification imposed by automata. Hence, we start from $k$ and go on asserting the behaviour imposed by automata $\re_g, \re_{g+1}, \ldots, \re_{g+h}$. Note that these automata do no allow occurrence of any oversampling point. Hence, assertion of these automata implies that all the points from $i'_g$ to $i'_{g+h}$ are within $(\tau'_k, \tau'_k+1]$. Finally we stop at point $i'_{g+h}$ from where we continue asserting first part of $\re_{g+h+1}$ till the next point labelled $j_{g+1}$\footnote{$-$ used in the subscript of $\nt$ is assumed to be - modulo $\cmax$.}.
 
$\psi_{1,\tseq, \dseq, k,g,h} {=}\psi_{2, \tseq, \dseq, k,g,h}= \psi_{\tseq, \dseq, k,g,h} =
\\\fut_{I_k}[\fregm(\re_{g+k+1}\cdots\re_{g+h}\cdot\re_{g+h+1,q_{g+h+1},1}\cdot \{\nt_{\tseq[g]}\})(S \cup \{\nt_{\tseq[g]}\}) \wedge\\ \sregm(\rev(\re_{g+k}).\rev(\re_{g+k-1})\cdots.\rev(\re_{g+1}).\rev(\re_{g,\dseq[g],2}).\{\nt_{\tseq[g]-1}\})\\(S \cup \{\nt_{\tseq[g]-1}\})]$
 
 Given a $\tseq$, let $P$ be the set of indices of $\tseq$ such that for any $p {\in} P$, $\tseq[p]$ occurs exactly once in $\tseq$. Similarly, let 
 $G'{=}\{(g,h) | \tseq[g] \ne \tseq[g+1]{=}\tseq[g+1]{=}\ldots{=}\tseq[g+h+1] \ne \tseq[g+h+2]\}$. $G_0{=}\{(g,h) {\in} G'|  u_g > \tseq[h] - \tseq[0]\%\cmax \wedge l_{g+h} \le  \tseq[h] - \tseq[0]\%\cmax\}$ and $G{=}G' \setminus G_0$. $P$ are the set of points whose corresponding intervals do not overlap with any other interval (hence formulae of non-overlapping case is applicable for these points). $G_0$ corresponds to the Case 0 (neither (A) nor (B)) of the Overlapping Case. $G$ corresponds the Case 1 and 2 (either (A) or (B) holds) of the Overlapping Case.
 Hence the final formula for given $\tseq$ and $\dseq$ as follows:
 $\psi_{\tseq, \dseq}{=}\fut_{[0,1)} \tseq[1] \wedge \bigwedge \limits_{1\le g <n \wedge g {\in} P} \psi^+_{g, \tseq[g+1],\dseq[g+1]} \wedge \bigwedge \limits_{1\le g \le n \wedge g {\in} P}\psi^-_{g, \tseq[g],\dseq[g]} \wedge \psi^+_n \wedge \\ \bigwedge \limits_{(g,h) {\in} G_0} \psi_{0, \tseq, \dseq,g,h} \wedge  \bigvee\limits_{(g,h) {\in} G} \psi_{\tseq, \dseq,g,h}$
The final required formula is disjunction over all $\tseq$ and $\dseq$. 
$$\psi'_f{=}\bigvee \limits_{\tseq {\in} \Tseq, \dseq {\in} \Dseq} \psi_{\tseq, \dseq} \wedge \varphi_{\ovs}$$

\subsection{Wrapping Up}
\noindent \textbf{4) Converting the $\emitl$ to $\emitl_{0,\infty}$}: While all the $\freg$ and $\sreg$ modalities are untimed, there are timing intervals of the form $\langle l, u \rangle$ where $l>0$ and $u \le \cmax$ associated with $\fut$ modality in $\psi'_f$. We can reduce these timing intervals into purely lower bound ($\langle l, \infty) $) or upper bound constraint ($\langle 0, u \rangle$) using these oversampling points, preserving equivalence by technique showed in \cite{khushraj-thesis} Chapter 5 lemma 5.5.2 Page 90-91. We present a note of intuition for the sake of completeness. After this step we finally get an $\emitl_{0,\infty}$ formula $\psi_f$. 
subsection{Converting timing constraints to $0,\infty$}
\label{app:0inf}
Consider subformulae of $\psi_f$ of the form $\fut_{\langle l,u \rangle}(\beta)$ where $l>0$ and $u$ is finite.
Let $\rho,1 \models \varphi_\ovs$. $\rho = (a_1, \tau_1)\ldots (a_n, \tau_n)$. Then $\rho,i \models \fut_{[ l,u )}(\beta) \iff \exists j > i. \tau_j - \tau_i \in [ l, u ) \rho,j \models (\beta)$. Note that, if $i \models \fut_{[0,1)} (\nt_k)$ then 
$\rho, i \models \psi_{r,k} = \fut_{[0,u)}\wedge \fut_{[0,1)}(\nt_{k\oplus u}) \wedge \beta$ iff there exists $j$ where $\beta$ holds and $\tau_j \in (\lfloor \tau_i + u \rfloor, \rho[i](2) + u )$. Moreover, if $\rho, i \models \psi {l, k} = \neg \nt_{k \oplus (u -1)} \until_{[l, \infty)} (\nt_{k \oplus (u -1)} \wedge \beta$ iff there exists a point $j$ such that $\tau_j - \tau_i \in [l \infty)$ but the point $j$ occurs before the next occurrence of $\nt_{k \oplus (u-1)}$ which occurs at timestamp $\lfloor \tau_i + u \rfloor$. Hence, $\rho,i \models \fut_{[ l, u )} \beta \wedge \varphi_{\ovs} \iff \rho,i \models \psi_{\beta} = \bigvee \limits_{g = 0}^{\cmax}(\fut_{[0,1)}\nt_g \wedge \psi_{l,g} \wedge \psi_{r,g}) \wedge \varphi_{\ovs}$. Hence, $\psi_\beta$ is equivalent to $\fut_{\langle l,u \rangle}(\beta) \wedge \varphi_{\ovs}$. Hence, every bounded timing constraint with $l > 0$ appearing in $\psi'_f$ can be replaced to an $(\emitl)_{0,\infty}$ formula $\psi_f$  such  that $|\psi_f| = \mathcal{O}(\cmax \times \psi'_f)$ (we assume binary encoding of constants appearing in the interval).
\subsection{EXPSPACE Upper Bound for non-adjacent $\pnregmtl$ and non-adjacent 1-$\tptl$ } 
\label{app:comp}
We first compute the size of $\emitl_{0,\infty}$ formula $\psi$ constructed from $\phi$ in the previous section.
Each $T_i$ can be replaced by $T'_i{=}\sbf(b_i\leftrightarrow \psi_f) \wedge \varphi_{\ovs}$ to get formula $\psi$ from $\phi_{flat}$.
Note that $|\psi_{\tseq, \dseq}|{=}\mathcal{O}[2n\times$(number of $\freg, \sreg$ formulae for each $\tseq, \dseq$) $ \times |\re_{max}|$(max size of each $\freg$, $\sreg$ formula)$]$ . $|\Tseq|{=}\mathcal{O}(\cmax^n)$, $|\Dseq|{=}m_1 \times m_2 \ldots m_{n+1}$. Each $m_i{=}\mathcal{O}(|\re_i|)$.Hence $|\Dseq|{=}\mathcal{O} |\re_{max}|^n$,\\$|\varphi_{\ovs}|{=}\mathcal{O}(\cmax)$,  $|\psi_f|{=}(|\psi_{\tseq, \dseq}| \times |\Tseq| \times |\Dseq|) + |\psi_{\ovs}|{=}\mathcal{O}(n \times (\cmax)^n \times |\re_{max}|^n$.
$|\psi'_f|{=} \mathcal{O}(n \times (\cmax)^n \times |\re_{max}|^n) = \mathcal{O}(2^{Poly(|\phi|)}$, where $n$ is the arity of $\phi_f$, $\cmax$ is the max constant used in the timing intervals f $\phi_f$. Final $\emitl_{0,\infty}$ formula $\psi_f$ is such that $|\psi_f| = \mathcal{O}(\cmax \times \psi'_f) = \mathcal{O}(2^{Poly(|\phi|)}$. Finally $|\psi| = \mathcal{O}(|\psi_f| \times |\phi|) = \mathcal{O}(2^{Poly(|\phi|)}$.
$\emitl_{0,\infty}$ is PSPACE complete.  Hence, we have an EXPSPACE procedure for non-adjacent $\pnregmtl$.

\end{document}